\documentclass[sigconf]{acmart}
\usepackage{multirow}
\usepackage{listings}
\usepackage{amsthm}
\theoremstyle{remark}
\newtheorem{remark}{Remark}
\newtheorem{observation}{Observation}
\usepackage{soul}
\usepackage[ruled,vlined,linesnumbered]{algorithm2e}
\usepackage{balance}
\lstdefinestyle{mypython}{
language=Python,
basicstyle=\ttfamily\small,
keywordstyle=\color{purple}, % 关键词（如库导入）为紫色
stringstyle=\color{orange},  % 字符串为橙色
commentstyle=\color{red},    % 注释为红色
showstringspaces=false,
breaklines=true,
frame=leftline,
numbers=left,
numberstyle=\tiny\color{gray}
}

\AtBeginDocument{%
}

\copyrightyear{2026}
\acmYear{2026}
\setcopyright{cc}
\setcctype{by}
\acmConference[WWW '26]{Proceedings of the ACM Web Conference 2026}{April 13--17, 2026}{Dubai, United Arab Emirates}
\acmBooktitle{Proceedings of the ACM Web Conference 2026 (WWW '26), April 13--17, 2026, Dubai, United Arab Emirates}
\acmPrice{}
\acmDOI{10.1145/3774904.3792385}
\acmISBN{979-8-4007-2307-0/2026/04}
\begin{document}


\title{FS\_GPlib: Breaking the Web-Scale Barrier — A Unified Acceleration Framework for Graph Propagation Models}

\author{Chang Guo}
\orcid{0009-0001-5525-0790}
\affiliation{%
	\institution{University of Science and Technology of China}
	\city{Hefei}
	\state{Anhui}
	\country{China}
}
\email{guochangcc@mail.ustc.edu.cn}

\author{Juyuan Zhang}
\orcid{0009-0008-4910-9759}
\affiliation{%
	\institution{University of Science and Technology of China}
	\city{Hefei}
	\state{Anhui}
	\country{China}
}
\email{zhangjuyuan2020@mail.ustc.edu.cn}

\author{Chang Su}
\orcid{0009-0006-1895-9000}
\affiliation{%
	\institution{University of Electronic Science and Technology of China}
    %\department{Institute of Fundamental and Frontier Sciences}
	\city{Chengdu}
	\state{Sichuan}
	\country{China}
}
\email{suchang@std.uestc.edu.cn}

\author{Tianlong Fan}
\orcid{0000-0002-9456-6819}
\authornote{Corresponding author}
\affiliation{%
	\institution{University of Science and Technology of China}
    %\department{School of Cyber Science and Technology}
	\city{Hefei}
	\state{Anhui}
	\country{China}
}
\email{tianlong.fan@ustc.edu.cn}

\author{Linyuan Lü}
\orcid{0000-0002-2156-0432}
\authornotemark[1]
\affiliation{%
	\institution{University of Science and Technology of China}
    %\department{School of Cyber Science and Technology}
	\city{Hefei}
	\state{Anhui}
	\country{China}
}
\email{linyuan.lv@ustc.edu.cn}
% \authornotetext[1]{Corresponding authors.}


\begin{abstract}
	
	%ver2, 200 words
	Propagation models are essential for modeling and simulating dynamic processes such as epidemics and information diffusion. However, existing tools struggle to scale to large-scale graphs that emerge across social networks, epidemic networks and so on, due to limited algorithmic efficiency, weak scalability, and high communication overhead. We present \textbf{FS\_GPlib}, a unified library that enables efficient, high-fidelity propagation modeling on Web-scale graphs. FS\_GPlib introduces a dual-acceleration framework: it combines micro-level synchronous message-passing updates with macro-level batched Monte Carlo simulation, leveraging high-dimensional tensor operations for parallel execution. To further enhance scalability, it supports distributed simulation via a novel target-node-based graph partitioning strategy that minimizes communication overhead while maintaining load balance. Theoretically, we show that under ideal assumptions, the runtime of simulations converges approximately to a constant. Extensive experiments demonstrate up to \textbf{35,000× speedup} over standard libraries such as NDlib and execution of a full Monte Carlo simulation on a Web-scale (billion-edge) graph in \textbf{11 seconds} while maintaining high simulation fidelity. FS\_GPlib supports 29 propagation models—including epidemic and opinion dynamics and dynamic network models—and offers a lightweight Python API compatible with mainstream data science ecosystems. By addressing the unique challenges of modeling diffusion and cascades on the Web, FS\_GPlib provides a scalable, extensible, and theoretically grounded solution for large-scale propagation analysis in epidemiology, social media analysis, and online network dynamics. Code available at: \url{https://github.com/Allen-Ciel/FS_GPlib}.

\end{abstract}





\maketitle

\newcommand\webconfavailabilityurl{https://doi.org/10.5281/zenodo.18315461}
\ifdefempty{\webconfavailabilityurl}{}{
\begingroup\small\noindent\raggedright\textbf{Resource Availability:}\\
% please change the following context to include multiple artifacts if necessary, including data, models, code, etc.
The source code of this paper has been made publicly available at \url{\webconfavailabilityurl}.
\endgroup
}

\section{Introduction}

Propagation models are fundamental tools for understanding dynamic processes on complex networks, from epidemic outbreaks to information diffusion and opinion dynamics. {On the Web—across social networks, microblogs, forums, and content-sharing platforms—these processes govern cascades, virality, and exposure.} Accurate simulation of these models enables effective analysis and informed decision-making across various fields such as epidemiology \cite{ghosh2022dynamics, jain2023optimal}, social media analysis \cite{ling2023deep, devarapalli2024estimating}, and financial risk assessment \cite{haldane2011systemic,acemoglu2015systemic}.

However, graphs routinely reach billions of edges on the Web, efficiently simulating propagation processes at this scale becomes computationally prohibitive. Existing simulation methods suffer inherent trade-offs between accuracy, efficiency, and scalability. Classical propagation models, such as Independent Cascade (IC) \cite{goldenberg2001using, goldenberg2001talk}, Susceptible-Infected-Recovered (SIR) \cite{pastor2015epidemic}, and Linear Threshold (LT) \cite{granovetter1978threshold}, entail repeated edge traversals and state updates, severely limiting their efficiency on large-scale data. Approximation techniques such as mean-field \cite{jenness2018epimodel} or $\tau$-leaping \cite{gillespie2001approximate} improve speed but sacrifice fidelity. Massively parallel methods \cite{koster2022gpu, barrett2008episimdemics} alleviate computational bottlenecks but still rely heavily on iterative node-level processing and fail to overcome fundamental scalability limitations. Integrated propagation algorithm libraries like NDlib \cite{rossetti2018ndlib}, EoN \cite{miller2020eon}, and CyNetDiff \cite{robson2024cynetdiff} provide convenient interfaces but do not fundamentally resolve these efficiency and scalability constraints. General-purpose systems for large-scale graph analytics (e.g., GraphLab \cite{gonzalez2012powergraph, gonzalez2014graphx}, Pregel \cite{malewicz2010pregel}, Giraph \cite{apacheGiraph}) provide scalable functions but lack propagation-specific optimizations.
Consequently, the fundamental performance bottleneck stemming from extensive iterative loops persists, posing significant challenges to efficient and accurate simulations on Web-scale networks.

To overcome these challenges, we propose \textbf{F}aster and more \textbf{S}calable python library for \textbf{G}raph \textbf{P}ropagation models (\textbf{FS\_GPlib}), a unified and general-purpose simulation library explicitly designed to break the Web-scale barrier in graph propagation modeling. 
FS\_GPlib adopts a dual-acceleration framework: (i) micro-level synchronous updates via message-passing computations reduce iterative operations; (ii) macro-level batched Monte Carlo simulations leverage tensor operations for parallel acceleration. Furthermore, a distributed scheme based on target-node partitioning balances workloads and minimizes inter-process communication. It provides 29 built-in models covering epidemic, opinion, and dynamic network scenarios, with an intuitive Python API for seamless integration into data science workflows. As an open-source, extensible, and well-documented framework, FS\_GPlib enables efficient large-scale propagation simulations with wide applicability in epidemiology, network science, social influence, and decision-making.

Extensive experiments demonstrate the superior performance and scalability of FS\_GPlib. Our framework completes a full Monte Carlo simulation on billion-edge networks within $\sim$\textbf{11 seconds}, marking a substantial advancement in propagation model simulations. It achieves up to \textbf{35,000× speedup} over popular standard libraries such as NDlib \cite{rossetti2018ndlib} on graphs with tens of millions of edges, while maintaining high simulation fidelity.

The remainder of this paper is organized as follows. Section~\ref{sec:preliminaries} reviews the fundamentals of graph structures and classical propagation models. Section~\ref{sec:theory} presents the theoretical foundations of our approach. Section~\ref{sec:architecture} details the architecture and API design of FS\_GPlib. Section~\ref{sec:experiment} reports extensive experiments validating the accuracy, efficiency, and scalability of our framework. Section~\ref{sec:summary} concludes the paper and outlines directions for future work.

\section{Preliminaries}\label{sec:preliminaries}

In this section, we introduce essential concepts related to graphs, the data format adopted in our framework, and the fundamental paradigm of mechanism-based propagation models.

\begin{definition}[Graph]
	A graph is $G=(V,E)$ with node set $V$ and edge set $E$. An edge $e \in E$ is either a tuple $(u,v)$ or a triple $(u,v,w)$, where $u$ and $v$ denote the source and target nodes, and $w$ is an optional edge weight. Graphs with weights are weighted; otherwise unweighted.
\end{definition}

\begin{definition}[Compressed Sparse Row, CSR]
	The CSR format compactly represents a sparse adjacency matrix by storing only nonzero elements. It consists of three arrays (row pointers, column indices, and values) of length proportional to $|E|$. For unweighted graphs, the values array can be omitted to save memory.
\end{definition}

We store graphs in CSR format, which integrates efficiently with message passing for node state updates. 
Compared to adjacency matrices (which suffer from quadratic space complexity) and adjacency lists (which do not support matrix operations), CSR offers a scalable compromise for large-scale simulations.

We now define the formal paradigm of \textit{mechanism-based propagation models}, which is central to this study.

\begin{definition}[Mechanism-Based Propagation Model]
	A propagation model governed by predefined rules and parameters, independent of historical data.
\end{definition}

These models originated from epidemiology and sociology~\cite{kermack1927contribution, daley1964epidemics}, with classical examples such as the SIR and IC models. More recent extensions incorporate probabilistic rules and structural heterogeneity~\cite{sun2020influence, li2022influence, kempe2003maximizing}.

\begin{definition}[Paradigm of Mechanism-Based Propagation Model]
	A mechanism-based propagation model consists of:
	\begin{itemize}
		\item A graph structure $G=(V,E)$.
		\item A node state space $\Sigma$.
		\item A probability space $\Omega$ representing all random variables and parameters (e.g., activation probabilities, thresholds).
		\item A discrete-time update sequence $\{(k, V_k)\ |\ k=0,1,\dots,K-1\}$, where $V_k \subseteq V$ is the set of nodes updated at step $k$, determined deterministically or sampled from $\Omega$.
		\item A node update function $F_v: \Sigma^{N^-(v)\cup \{v\}} \times \Omega \rightarrow \Sigma$, which maps the previous-round states of node $v$ and its in-degree neighbors $N^-(v)$ to the updated state.
	\end{itemize}
	
	Given initial states $\{X_{v,0} \in \Sigma\ |\ v \in V\}$ and random factors $\omega \in \Omega$, the propagation process proceeds iteratively. At each step $k$, for each $v \in V_k$, the state is updated as:
	\begin{equation}\label{eq:nodestate}
		X_{v,k} = F_v(X_{v,k-1}, \{X_{u,k-1}: u \in N^-(v)\}, \omega).
	\end{equation}
	States of nodes not in $V_k$ remain unchanged.
\end{definition}

The implementation of mechanism-based propagation models typically involves four nested loops (Algorithm ~\ref{code}). The outermost loop involves Monte Carlo simulations to estimate outcomes under probabilistic conditions $\omega\in\Omega$. Within each simulation, the model iterates over $K$ discrete steps, updating designated nodes at each step. Updating a node further involves iterating over its neighbors or edges to compute their influence. As a result, increasing the number of simulations, time steps, nodes, or edges leads to a significant increase in computational cost.

\begin{algorithm}
{\color{black}\caption{Traditional mechanism-based propagation}
\label{code}
\KwIn{$G=(V,E)$, time steps $K$, Monte Carlo simulations $S$}
\KwOut{$\mathcal{X}^K=\{X_s^K\}_{s=1}^S$}
Initialize an empty container $\mathcal{X}^K$\;
\For{$s \gets 1$ \KwTo $S$}{
    Initialize states of nodes $X^0 = \{x_i^0 \mid i \in V\}$\;
%    $k \gets 0$\;
    \For{$k \gets 0$ \KwTo $K-1$}{
        $X^{k+1} \gets X^k$\;

        \For{$i \in V$}{
            Decide whether $x_i^k$ needs to update\;
            \If{update is required}{
                Use $\{x_j^k \mid j \in N(i)\}$ or a random probability\;
                Update $x_i^{k+1}$ accordingly\;
            }
            \Else{
                $x_i^{k+1} \gets x_i^k$\;
        }
    }
    Record $X^{k+1}$\;
%        $k \gets k+1$\;
    }
    Append $X^K$ to $\mathcal{X}^K$\;
}
\Return $\mathcal{X}^K$
}
\end{algorithm}

\section{Theoretical Foundations of FS\_GPlib}\label{sec:theory}

This section details the theoretical grounding of the dual-accelerated propagation framework and distributed simulation methods introduced by FS\_GPlib. We illustrate their implementations via two representative propagation models: the SIR and Hegselmann-Krause (HK) models in Appendix ~\ref{SIRHK}.

\subsection{Dual-Acceleration Framework via Message Passing}

Propagation models typically require numerous time steps and Monte Carlo runs for tracking node states over time and estimating outcomes. FS\_GPlib addresses computational bottlenecks through a dual-acceleration strategy that combines micro-level synchronous node-state updates and macro-level batched Monte Carlo computations (Fig.~\ref{fig:MessagePassing}).

\begin{figure}[h]
	\centering
	\includegraphics[width=\linewidth]{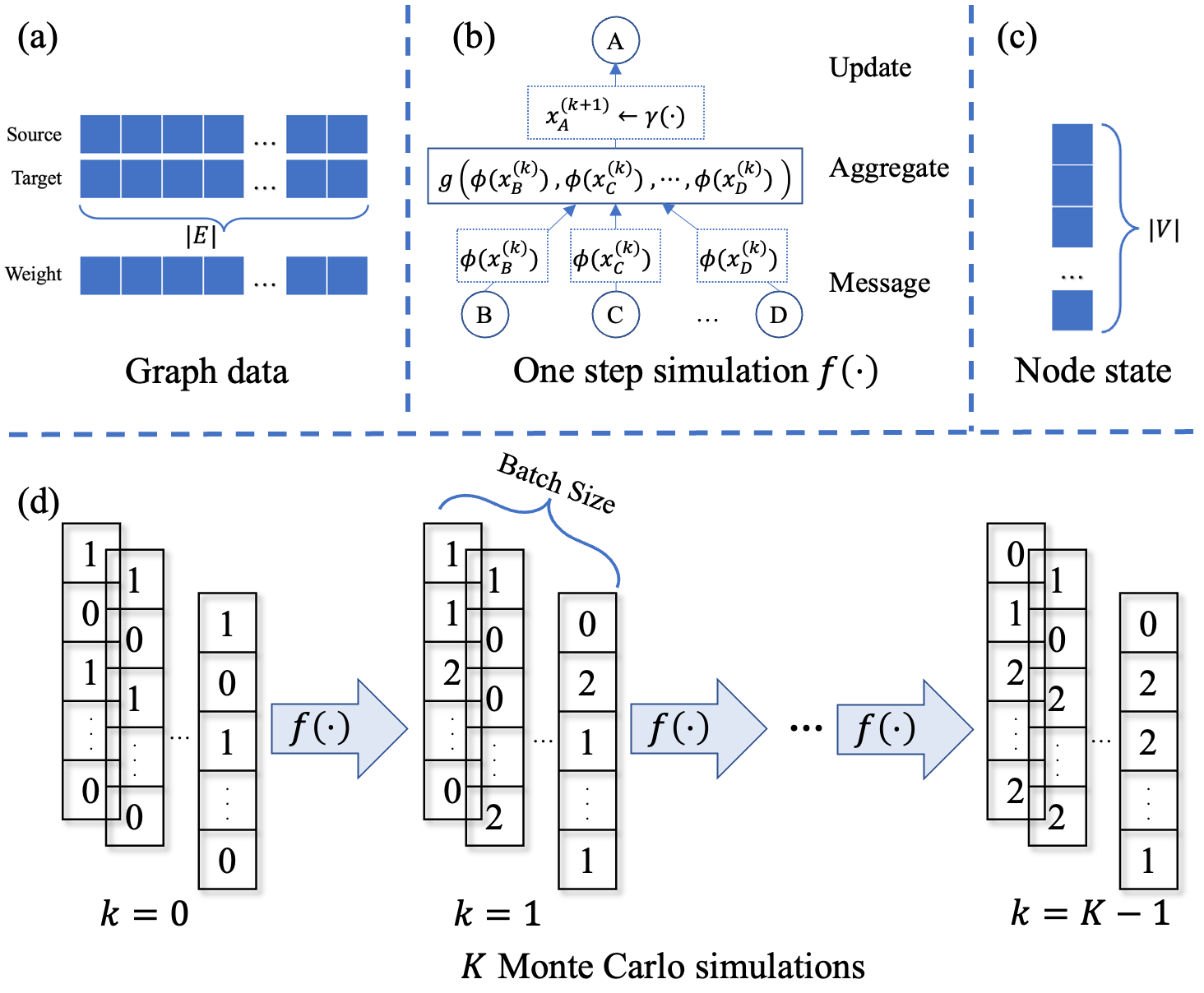}
	\caption{Architecture of the dual-acceleration propagation framework based on micro-level message passing and macro-level batched Monte Carlo computations.}
	\Description{}
	\label{fig:MessagePassing}
\end{figure}

\paragraph{Micro-level Acceleration via Node-Synchronous Computation}
Traditional propagation algorithms repeatedly traverse nodes and their neighbors, resulting in significant computational redundancy. To implement efficiently, we rely on the following observation about the locality of node activation.

\begin{observation}

In the mechanism propagation model, the evolution pattern of nodes is similar to the forward propagation paradigm of the Message Passing framework, including message passing, aggregation, and state updating (Fig.~\ref{fig:MessagePassing}(b)).
\end{observation}

We propose a micro-level node-synchronous computation method that aggregates neighbor information for every node at once and updates all node states simultaneously, thereby eliminating redundant traversals and enhancing computational efficiency.

\begin{definition}[Node-Synchronized Update Method]
	Let $\boldsymbol{x}_i^{(k)}\in\Sigma$ denote the state of node $i$ at step $k$, and $\omega^{(k)}$ represent the random component at step $k$. At each time step, source nodes transmit messages to destination nodes, which then aggregate the received messages and update their states as:
	\begin{equation}
		\boldsymbol{x}_i^{(k)}=\gamma\left(\boldsymbol{x}_i^{(k-1)},\underset{j\in\mathcal{N}(i)}{g}\phi(\boldsymbol{x}_i^{(k-1)},\boldsymbol{x}_j^{(k-1)},\boldsymbol{e}_{j,i}, \omega^{(k)}), \omega^{(k)}\right),
		\label{eq:all}
	\end{equation}
	where functions $\phi(\cdot)$, $g(\cdot)$ and $\gamma(\cdot)$ encode message passing, aggregation, and node-state updating, respectively.
\end{definition}

In propagation models, nodes are represented by one‐dimensional features, and some state transitions depend solely on probabilities rather than neighbor states. Therefore, in simulating a single time step of the propagation model, the states of neighboring nodes are aggregated first:
\begin{equation}
	\boldsymbol{m}_{ji}^{\left(k\right)}=\phi\left(\boldsymbol{x}_i^{\left(k-1\right)},\boldsymbol{x}_j^{\left(k-1\right)},\boldsymbol{e}_{j,i},  \omega^{(k)}\right),
	\label{eq:ms}
\end{equation}
function $\phi(\cdot)$ is defined according to the propagation‐model mechanism and typically takes edge information, the states of the source and target nodes from the previous time step, and the random element sampled at this step or initially. Then aggregate the received information:
\begin{equation}
	\boldsymbol{m}_i^{\left(k\right)}=g\left(\{\boldsymbol{m}_{ji}^{\left(k\right)}\middle| j\in\mathcal{N}\left(i\right)\}\right),
	\label{eq:agg}
\end{equation}
the aggregated information $\boldsymbol{m}_i^{\left(k\right)}$ is used to update the node status:
\begin{equation}
	\boldsymbol{x}_i^{\left(k\right)}=\gamma\left(\boldsymbol{x}_i^{\left(k-1\right)},\boldsymbol{m}_i^{\left(k\right)},  \omega^{(k)}\right),
	\label{eq:updata}
\end{equation}
during the node state update step, the update method for each target node is determined based on its current state $\boldsymbol{x}_i^{\left(k-1\right)}$. It may be updated using neighbor information $\boldsymbol{m}_i^{\left(k\right)}$, or, if its update is independent of neighbors, via other parameters.

Micro-level acceleration is achieved by performing the above operations within tensor computations to compute all nodes in parallel, reducing interpreter overhead and exploiting CPU/GPU parallelism.

\paragraph{Macro-level Acceleration via Batched Monte Carlo Simulations}
Monte Carlo simulations repeatedly execute identical algorithms on the same graph. Inspired by deep learning batch processing, we propose:

\begin{definition}[Batched Monte Carlo Method]
	Node states from $B$ Monte Carlo simulations are stacked into a tensor $\mathcal{X}\in\mathbb{R}^{B\times|V|}$, enabling parallel propagation computations:
	\begin{equation}
		f(\mathcal{X},\{\omega\}^B),
		\label{eq:batch}
	\end{equation}
	where each simulation within the batch independently samples randomness $\omega$.
\end{definition}

Figure \ref{fig:MessagePassing}(d) illustrates the batch-parallel method. In this approach, multiple initial node-state vectors are stacked into a high-dimensional tensor, and forward propagation is executed as a single batched operation across $K$ discrete time steps. This enables simultaneous execution of multiple Monte Carlo simulations through vectorized operations.

Compared to conventional multi-process approaches—which distribute Monte Carlo simulations across CPU cores—batch parallelism consolidates multiple node states into a single high-dimensional tensor, enabling all runs to be executed concurrently within one forward-propagation pass. This design avoids repeated model instantiation and redundant operator invocations, yielding higher computational efficiency. Proper tuning of the batch size also ensures better GPU utilization.

By integrating both micro-level (synchronous node updates) and macro-level (batched Monte Carlo) acceleration, the proposed framework significantly reduces simulation runtime while preserving the time-step iteration logic. This enables fine-grained tracking of node-state evolution over time and efficient estimation of the expected propagation range.

\subsection{Distributed Simulation via Target-node-based Partitioning}\label{subs:partition}

Simulating propagation on Web-scale graphs poses significant challenges in memory consumption and computational efficiency. Partitioning strategies are typically categorized into \emph{vertex partitioning} and \emph{edge partitioning}. {Vertex partitioning methods, such as METIS~\cite{karypis1998fast}, aim to minimize cross-partition communication, while edge partitioning emphasizes load balancing and reducing vertex replication.

FS\_GPlib adopts a distributed simulation strategy that partitions the graph across multiple processes based on target nodes, aiming to balance computational loads while minimizing communication overhead (Fig.~\ref{fig:Distributed}). In propagation models, target nodes aggregate messages from source nodes to update their state. If a target spans multiple partitions, inter-process communication is needed to exchange partial messages, adding cost and complexity. FS\_GPlib avoids this by assigning each target node to exactly one partition, enabling local aggregation and updates without cross-partition messaging. After local updates, an \texttt{all\_reduce} synchronizes target-node states across processes, and the updated global state is broadcast for the next iteration—ensuring correctness while minimizing redundant communication.

\begin{figure}[h]
	\centering
	\includegraphics[width=\linewidth]{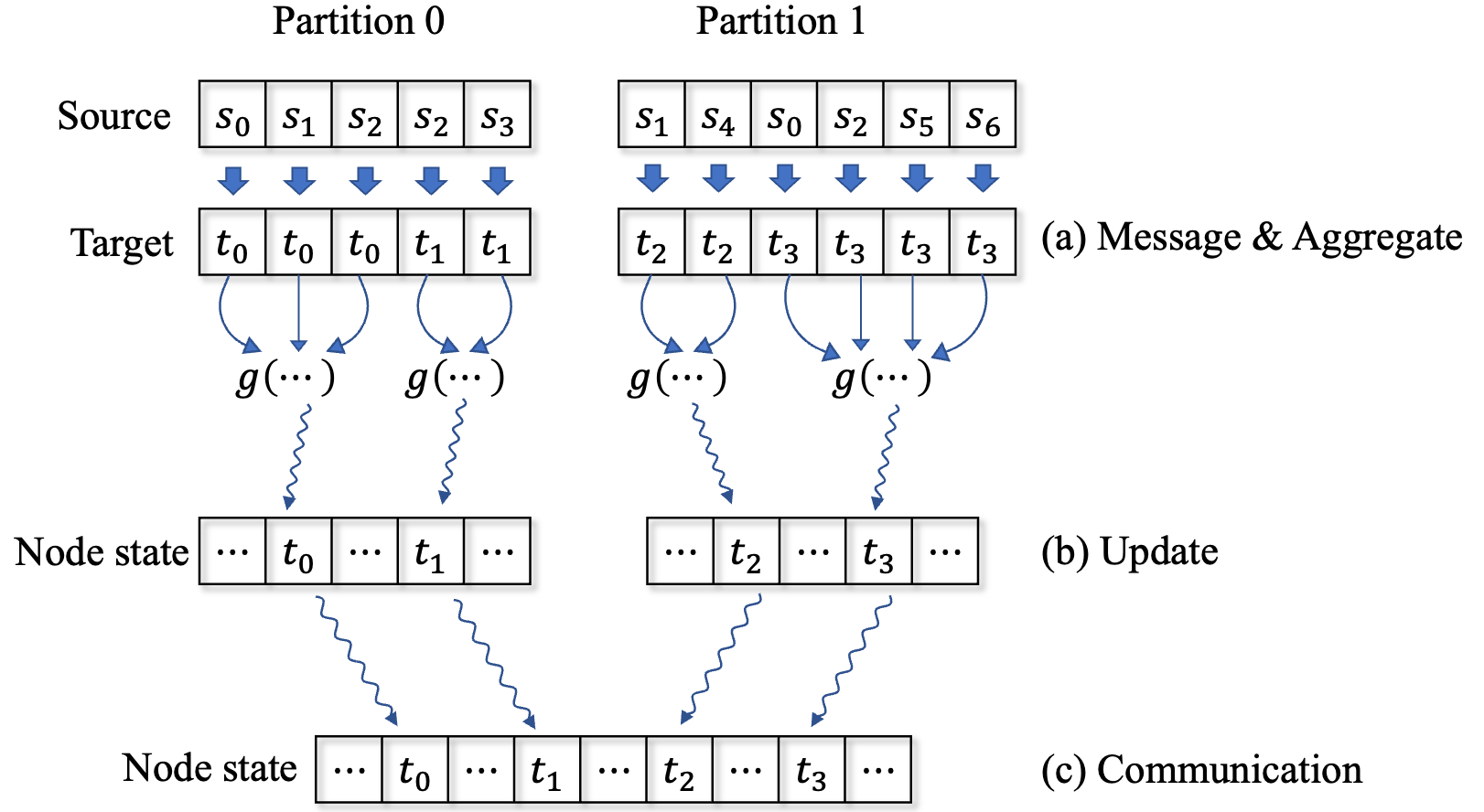}
	% \caption{Distributed propagation simulation method.}
	\caption{Distributed propagation simulation strategy with target-node-based partitioning.}
	\Description{}
	\label{fig:Distributed}
\end{figure}

\begin{definition}[Target-Node-Based Partitioning via Longest Processing Time Heuristic]
	Sort nodes by non-increasing in-degree to obtain the permutation $\pi=(\pi_1,\pi_2,\dots,\pi_N)$. 
	Let the prefix sum of weights be $C_i=\sum_{n=1}^i w_{\pi_n}$ and the total weight $W=C_N$. 
	Each node $\pi_i$ is assigned to partition
	\begin{equation}
		q(i)=\left\lceil \tfrac{d\,C_i}{W}\right\rceil,\qquad q(i)\in\{1,\dots,d\},
	\end{equation}
	where $d$ is the number of partitions. Define the partitioned target-node sets and induced edge sets by
	\begin{equation}
		P_q=\{\pi_i \mid q(i)=q\}, \qquad 
		E_q=\{(u,v)\in E \mid v\in P_q\}.
	\end{equation}
	
	This guarantees every target node belongs to exactly one partition, and each process handles a subgraph with approximately balanced edge load.
\end{definition}

{\color{black}\begin{proposition}[Load Balance Guarantee]
	Let $w_{\max}=\max_j w_j$. For each partition $q$, the edge load satisfies
	\begin{equation}
	\left|\sum_{i\in P_q} w_{\pi_i}-\frac{W}{d}\right|< w_{\max},
	\end{equation}
	thus the workload deviation from the ideal average $W/d$ is bounded by at most the maximum in-degree.
\end{proposition}}

\begin{remark}[Communication Efficiency]
	Since each target node appears in exactly one partition, its state can be updated locally without inter-process communication. Consequently, at each time step, only the updated target-node states need to be communicated across processes, yielding a communication complexity of $O(|V|)$ per step, independent of the total number of edges $|E|$.
\end{remark}

\subsection{Theoretical Analysis}\label{sec:guarantee}

A key theoretical contribution of FS\_GPlib lies in its theoretically grounded ability to achieve near-constant runtime complexity through carefully designed micro-level synchronous node-state updates and macro-level batch-parallel Monte Carlo execution.

In line with most real-world graphs, we assume that the number of edges scales linearly with the number of nodes, i.e., $m = \Theta(n)$. We consider discrete-time propagation models, where each node updates its state in constant time based on neighbor states and random variables.

\begin{theorem}[Constant-Time Micro-Level Node Synchronization]\label{thm:constant_micro}
    Let $G=(V,E)$ be a sparse graph with $|V|=n$ nodes and $|E|=m=\Theta(n)$ edges. Consider a synchronous node-update rule (Eq.~\ref{eq:all}) under a message-passing-based propagation model, where each node's update is computed from its immediate neighbors' states or random variables.
    
    Assuming an ideal parallel execution environment where $p = \Theta(m)$ processing units can concurrently handle all $m$ edges without sequential fallback, and runtime overhead is negligible, the runtime per synchronous update step is:
    \begin{equation}
        T_{\text{step}} = O(1).
    \end{equation}

\end{theorem}

\begin{corollary}[Constant-Time Macro-Level Batched Simulation]\label{cor:macro_batch}
    Building on Theorem~\ref{thm:constant_micro}, suppose we execute $S$ independent Monte Carlo simulations simultaneously via batched tensor computations, where each simulation involves $K=O(1)$ propagation steps.
    
    Assuming an ideal execution environment where the device has sufficient capacity for $B=\Theta(S)$ copies of data to run simultaneously, and ignoring the runtime overhead, the overall time complexity is:
    \begin{equation}
    T_{\text{all}} = O(1).
    \end{equation}

\end{corollary}

\noindent See the Appendix~\ref{2proof} for the proofs of the above theorem and corollary.

These theoretical results explain FS\_GPlib's empirically observed high efficiency and scalability: with constant-time micro-level synchronization and macro-level batching, our framework presents constant time complexity on graphs below a certain size, and enables efficient simulation on billion-scale graphs.

\section{Architecture of FS\_GPlib}\label{sec:architecture}

This section introduces the architecture, usage, and customization capabilities of FS\_GPlib.

\subsection{Ensemble Algorithms}
\begin{table*}[htbp]
	\centering
	\caption{Summary of built-in propagation models in FS\_GPlib.}
	\small
	\begin{tabular}{llccc}
		\toprule
		\textbf{Model Class} & \textbf{Model Name} & \textbf{Weighted Graph Support} & \textbf{Iteration Mode}* & \textbf{State Types} \\
		\midrule
		\multirow{13}{*}{Epidemic Models}
		& SI & Yes & Synchronous Model & Discrete \\
		& SIS & Yes & Synchronous Model & Discrete \\
		& SIR & Yes & Synchronous Model & Discrete \\
		& SEIR(DT) & Yes & Synchronous Model & Discrete \\
		& SEIR(CT) & Yes & Synchronous Model & Discrete \\
		& SEIS(DT) & Yes & Synchronous Model & Discrete \\
		& SEIS(CT) & Yes & Synchronous Model & Discrete \\
		& SWIR & Yes & Synchronous Model & Discrete \\
		& Threshold & Yes & Synchronous Model & Discrete \\
		& Kertesz Threshold & Yes & Synchronous Model & Discrete \\
		& Independent Cascades & No & Synchronous Model & Discrete \\
		& Profile & No & Synchronous Model & Discrete \\
		& Profile Threshold & No & Synchronous Model & Discrete \\
		\midrule
		\multirow{6}{*}{Opinion Dynamics Models}
		& Voter & No & Asynchronous model & Discrete \\
		& Q-Voter & No & Asynchronous model & Discrete \\
		& Majority Rule & No & Asynchronous model & Discrete \\
		& Sznajd & No & Asynchronous model & Discrete \\
		& Weighted Hegselmann-Krause & No & Synchronous Model & Continuous \\
		& Hegselmann-Krause & No & Synchronous Model & Continuous \\
		\midrule
		\multirow{10}{*}{Dynamic Graph Models}
		& DySI & Yes & Synchronous Model & Discrete \\
		& DySIS & Yes & Synchronous Model & Discrete \\
		& DySIR & Yes & Synchronous Model & Discrete \\
		& DySEIR(DT) & Yes & Synchronous Model & Discrete \\
		& DySEIR(CT) & Yes & Synchronous Model & Discrete \\
		& DySEIS(DT) & Yes & Synchronous Model & Discrete \\
		& DySEIS(CT) & Yes & Synchronous Model & Discrete \\
		& DySWIR & Yes & Synchronous Model & Discrete \\
		& DyThreshold & Yes & Synchronous Model & Discrete \\
		& DyKerteszThreshold & Yes & Synchronous Model & Discrete \\
		\bottomrule
		\multicolumn{5}{l}{
			\footnotesize        
			$*$: Synchronous/Asynchronous in the table refers to updating all/some nodes in one time step.
		}
	\end{tabular}
	\label{tab:29s}
\end{table*}

FS\_GPlib includes 29 built-in propagation models classified into three categories: Epidemic Models, Opinion Dynamics Models, and Dynamic Graph Models (Table \ref{tab:29s}). All models support both directed and undirected graphs, with certain models additionally supporting weighted graphs. Each model is characterized by two primary dimensions: iteration mode (synchronous vs. asynchronous) and state representation (discrete vs. continuous). In synchronous iteration, all node states are updated in a single step, whereas asynchronous iteration updates one node or subset per step. Discrete models restrict node states to finite sets, while continuous models allow states within continuous intervals.

\subsection{Algorithm Usage}

\subsubsection{Basic Workflow}

FS\_GPlib relies on the \textit{torch} and \textit{PyG} libraries, which must be installed beforehand. Using FS\_GPlib involves two steps: instantiating the model and running simulations.

To instantiate a propagation model $f(\cdot)$, specify the graph $data$, initial node set $seeds$, model parameters $\{params\}$, and computing device $device$:
\begin{equation}
    \text{Model} = f(data, seeds, \{params\}, device).
    \label{eq:model}
\end{equation}

Four simulation execution interfaces are available:
\begin{itemize}\label{4_interfaces}
    \item \textit{Model.run\_iteration()}: Execute one time step from the current node state and return the updated node state.
    \item \textit{Model.run\_iterations(times)}: Execute multiple time steps starting from the current node state and return the final node state after all iterations.
    \item \textit{Model.run\_epoch(times)}: Execute multiple time steps starting from the initial state and return the final node state after all iterations.
    \item \textit{Model.run\_epochs(epochs, times, batch\_size)}: Perform multiple Monte Carlo simulations in batches, each starting from the initial state, and return the final node states.
\end{itemize}

The provided interfaces collectively satisfy common requirements for propagation analysis. They support both fine-grained tracking of dynamics (e.g., per-iteration evolution) and large-scale statistical analysis of final outcomes (e.g., Monte Carlo estimation of spread range). We give an example of usage in the section ~\ref{example}.

%added
\subsubsection{Distributed Partitioning for {\color{black}Web}-scale Graphs}

FS\_GPlib supports {\color{black}Web}-scale graphs using target-node partitioning (Section \ref{subs:partition}), which balances computational load, minimizes communication overhead, supports partitioning of both edges and one-dimensional edge features, and stores subgraph metadata to a specified path.The workflow is as follows:
\begin{enumerate}
    \item Instantiate the partitioner:
    \begin{equation}
        \text{partitioner} = \text{GraphPartitioner}(data, n\_parts, root).
    \end{equation}
    
    \item Generate and save partitions:
    \begin{equation}
        \text{partitioner.generate\_partition()}.
    \end{equation}
    
    \item Load specific subgraph partitions:
    \begin{equation}
        \text{subdata, subnodes} = \text{load\_partition}(root, partition\_idx).
    \end{equation}
\end{enumerate}

If $root\neq None$, partitioned subgraphs and node lists are saved under the specified root path.

This partitioning strategy scales propagation to larger graphs without device constraints. Simulations may run sequentially on one node or be accelerated in parallel across multiple nodes. 
While FS\_GPlib does not provide a built-in interface for distributed simulation workflows, a reference implementation is available. The library also includes a suite of specialized parameters and utility functions to facilitate flexible customization.

\subsection{Custom Model Development}
FS\_GPlib includes three main classes of pre-implemented models, derived from the base classes \textit{DiffusionModel} and \textit{Diffusion\_process}. \textit{DiffusionModel} handles data loading, seed nodes, parameters, device configuration, and operational interfaces. \textit{Diffusion\_process} encapsulates propagation logic. Users can either directly use existing models or extend these base classes to create custom propagation algorithms tailored to specific needs.

\section{Experimental Evaluation}\label{sec:experiment}

We conduct comprehensive experiments to evaluate FS\_GPlib’s accuracy, efficiency, and scalability. The evaluation includes:
\begin{enumerate}
    \item \textbf{Experimental Setup}: Describes the datasets, models, parameters, software, and hardware used in the experiments.

    \item \textbf{Accuracy Validation}: Verifies the correctness of FS\_GPlib’s implementations by comparing simulation outcomes against established baselines.
    \item \textbf{Efficiency Benchmarking and Acceleration Analysis}: Assesses both micro-level and macro-level acceleration by comparing runtime across implementations and analyzing the effect of batch size on simulation throughput.
    \item \textbf{Scalability on Billion-Scale Graphs}: Demonstrates the distributed simulation capability of FS\_GPlib.
    
\end{enumerate}

\subsection{Experimental Setup}\label{subs:setup}

We select six real-world graph datasets spanning various sizes: Cora, PubMed, Flickr, Yelp, Pokec, and Webbase (Table~\ref{tab:dataset}). Among them, Pokec is used solely for FS\_GPlib-based simulation, and Webbase is used solely for distributed simulation.

\begin{table}[ht]
    \caption{Summary of experimental graph datasets.}
    \label{tab:dataset}
    
    \begin{tabular}{cccc}
        \toprule
        Dataset&	Nodes $(|N|)$ &  Edges $(|E|)$ & Avg. Degree $\langle k \rangle$ \\
        \midrule
        Cora~\cite{yang2016revisiting}
        &	2,708& 	10,556 &	7.80 \\
        PubMed~\cite{yang2016revisiting}
        &	19,717& 	88,648 &	8.99  \\
        Flickr~\cite{zeng2019graphsaint}
        &	89,250& 	899,756 &	20.16 \\
        Yelp~\cite{zeng2019graphsaint}
        &	716,847& 	13,954,819	&38.93 \\
        Pokec~\cite{takac2012data}
        &	1,632,803&	30,622,564&	37.51\\
        Webbase~\cite{boldi2004webgraph}	
        &118,142,143&	992,844,891	&17.27\\
        \bottomrule
    \end{tabular}
\end{table}

To evaluate the accuracy and efficiency of FS\_GPlib’s propagation models, we compare them against several widely used or highly efficient implementations of the IC and SIR models~\cite{youssef2011individual}: (1) a custom implementation using \textit{NetworkX} (Nx-*); (2) the \textit{NDlib} framework (Nd-*); (3) the \textit{CyNetDiff} library for IC~\cite{robson2024cynetdiff} (CND-IC); (4) a GPU-based parallel IC implementation~\cite{su2025gpic}(GP-IC); (5) the \textit{EoN} library for SIR~\cite{kiss2017mathematics} (EoN-SIR), which achieved top performance in recent studies~\cite{samaei2025fastgemf}; (6) a C++ implemented Python module \textit{Graph-tool}~\cite{peixoto_graph-tool_2014} for SIR (GT-SIR).

Unless otherwise specified, the IC model is run with infection probability $\beta = 0.5$ until convergence; the SIR model uses infection probability $\beta = 0.01$, recovery rate $\lambda = 0.005$, and runs for 100 time steps; all simulations use 1,000 Monte Carlo trials. Seed nodes are selected by degree centrality (top 10\%). 

Baseline implementations only support CPU execution, except for GPIC. FS\_GPlib supports both CPU and GPU. CPU tests were conducted on a compute node with four Intel Xeon Platinum 8260 CPUs (24 cores and 48 threads each, 2.40 GHz fundamental frequency, 1.5 TiB RAM). GPU tests used eight NVIDIA RTX 4090 GPUs (24.5 GiB each) on a dual-socket 64-core Intel Xeon Platinum 8358 (2.6 GHz) server with 503 GiB RAM. {\color{black}Unless otherwise specified, only one GPU is used by default.} The software environment includes Python 3.10, PyTorch 2.1.2, CUDA 12.1, with CUDA\_VISIBLE\_DEVICES=0. Detailed device-level runtime performance is presented in Section~\ref{app:devices}.

\subsection{Accuracy Validation}

This section verifies the accuracy of FS\_GPlib’s propagation models by comparing their outputs against established implementations. We focus on two widely studied models—IC and SIR—and defer additional results to Table~\ref{tab:result} in Appendix~\ref{app:accuracy_other_models}.

Figures~\ref{fig:IC_3} and \ref{fig:SIR_3} compare the Monte Carlo simulation results of FS\_GPlib with two widely used implementations: \textit{NetworkX} (Nx-*) and \textit{NDlib} (Nd-*). Due to performance constraints, we ran 10 simulations for Nd-*, 100 for Nx-*, and 1,000 for FS-* to ensure reliable convergence. For IC, Figure~\ref{fig:IC_3} presents the average number of activated nodes under varying propagation probabilities; for SIR, Figure~\ref{fig:SIR_3} shows the average percentage of infection nodes across datasets. In both cases, FS\_GPlib reproduces the correct dynamics: its mean trajectories align closely with those of other methods. The broader confidence bands in FS\_GPlib may be attributable to the increased number of simulation samples. These results validate the fidelity of FS\_GPlib's model implementations.

\begin{figure}[t]
    \centering
    \begin{minipage}{0.23\textwidth}
        \centering
        \includegraphics[width=\textwidth]{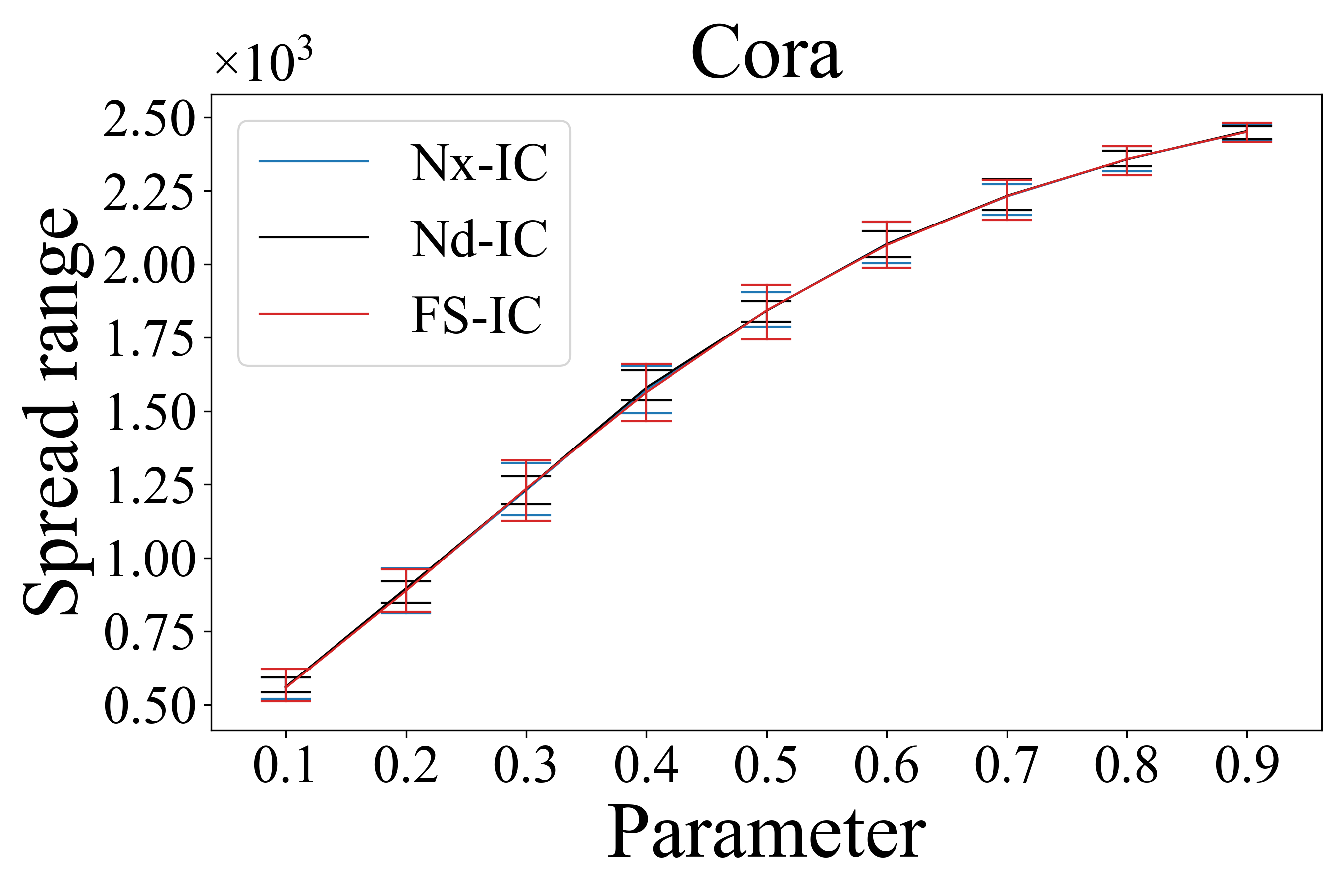}
        \small (a)
    \end{minipage}
    \begin{minipage}{0.23\textwidth}
        \centering
        \includegraphics[width=\textwidth]{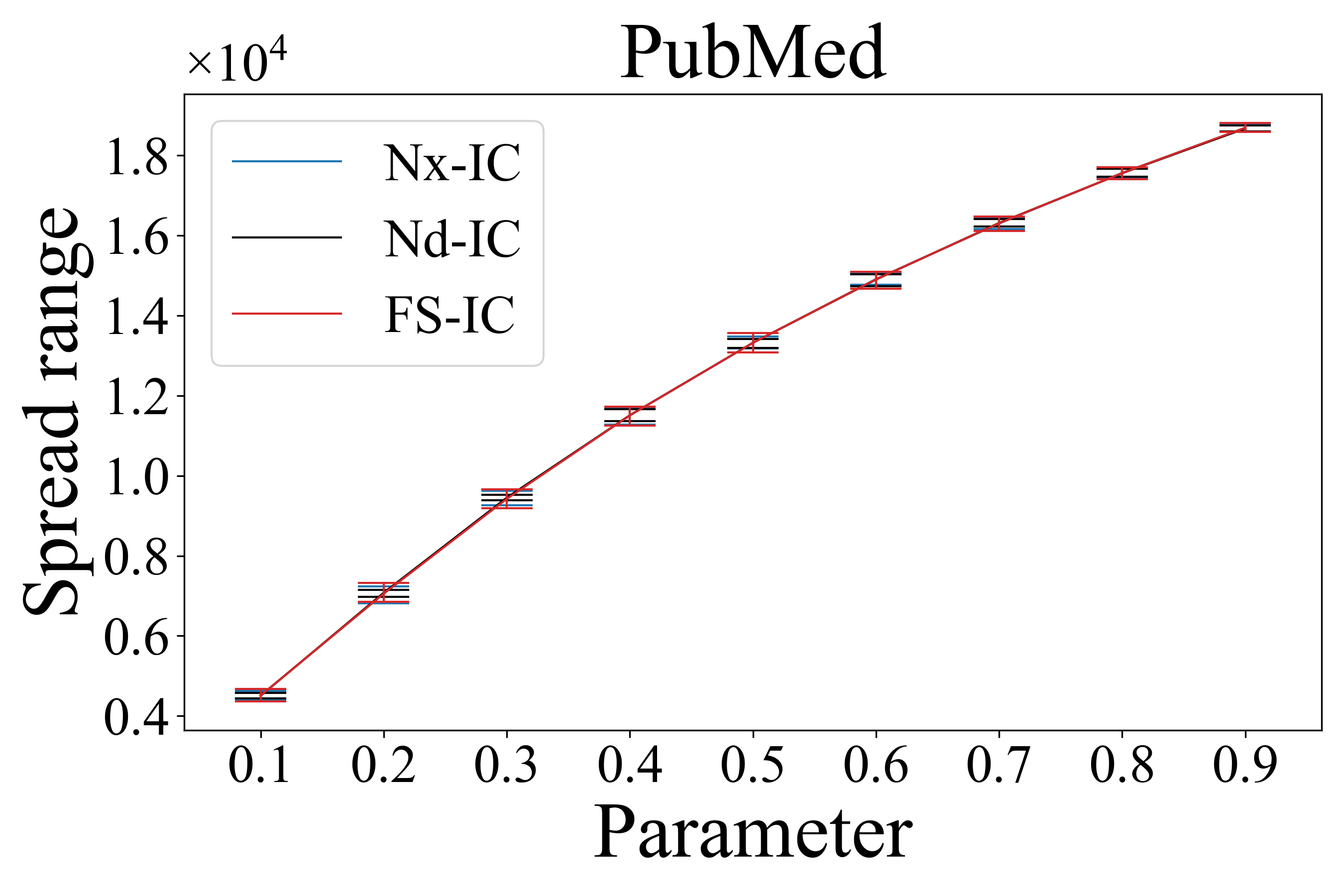}
        \small (b)  
    \end{minipage}
    
    \hspace{0.005\textwidth}
    
    \begin{minipage}{0.23\textwidth}
        \centering
        \includegraphics[width=\textwidth]{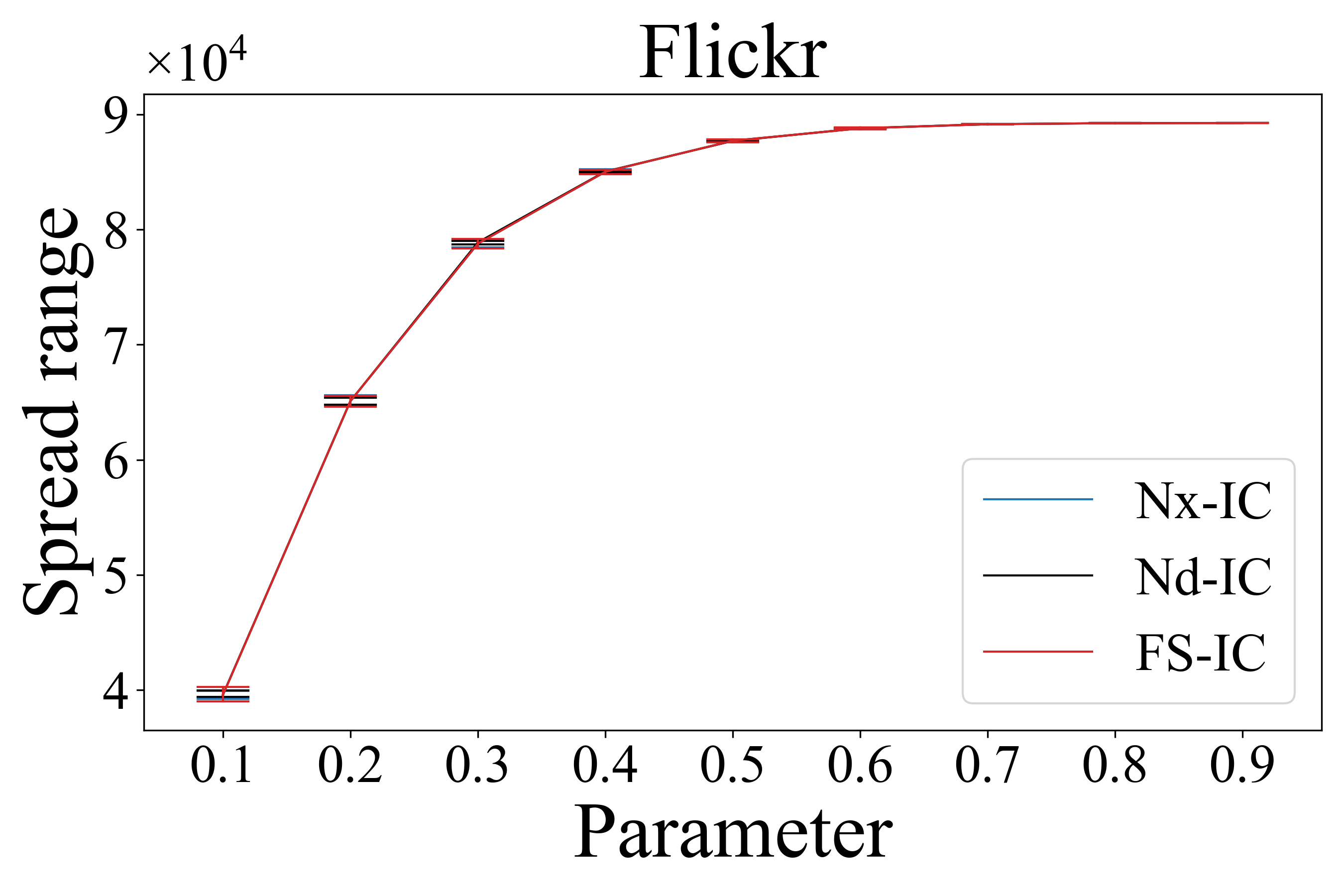}
        \small (c)
    \end{minipage}
    \begin{minipage}{0.23\textwidth}
        \centering
        \includegraphics[width=\textwidth]{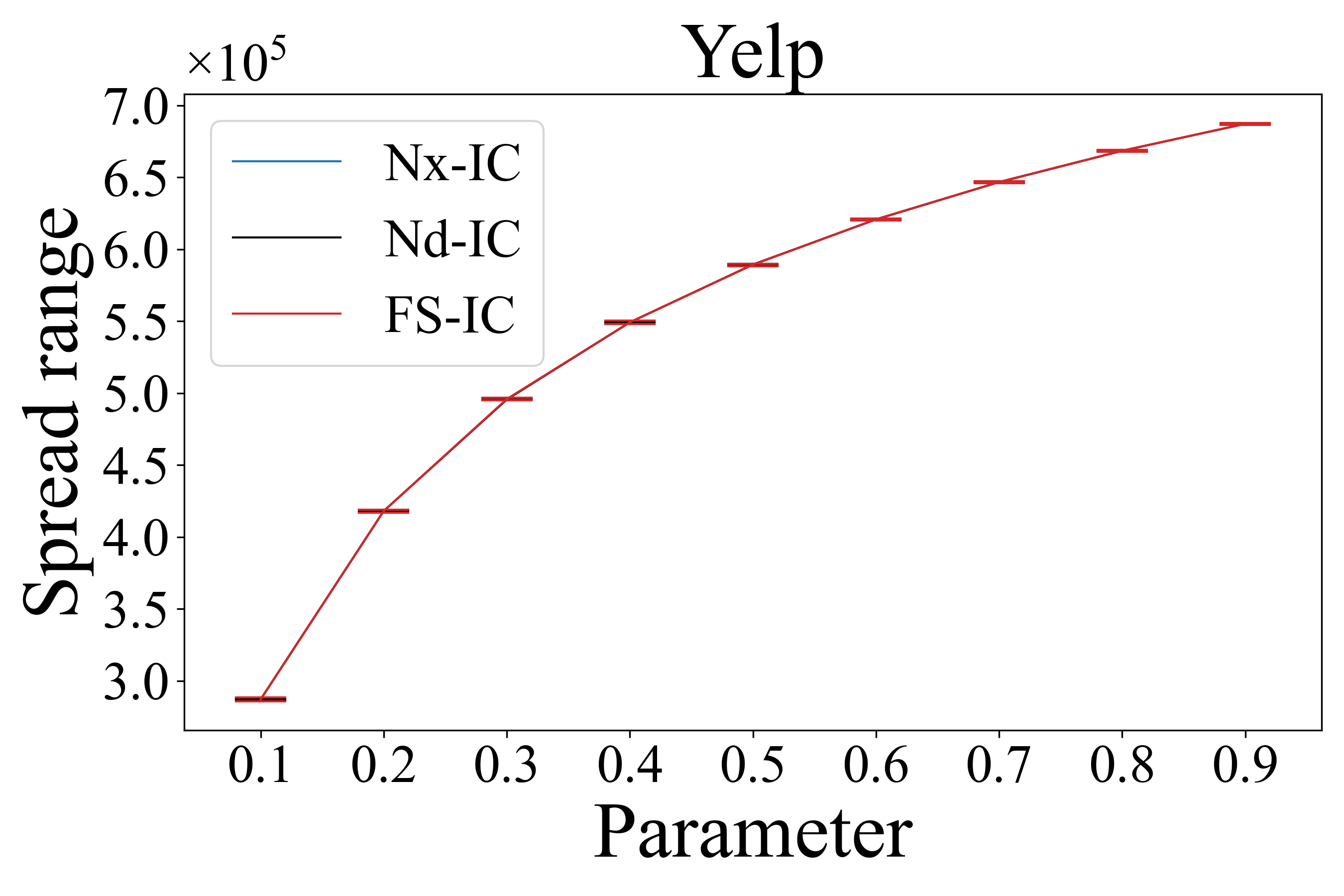}
        \small (d)
    \end{minipage}
    \caption{Monte Carlo simulation results for three IC implementations across four datasets, highlighting FS\_GPlib’s consistency and accuracy.}
    \label{fig:IC_3}
    \Description{}
\end{figure}

\begin{figure}[ht]
    \centering
    \begin{minipage}{0.43\textwidth}
        \centering
        \includegraphics[width=\textwidth]{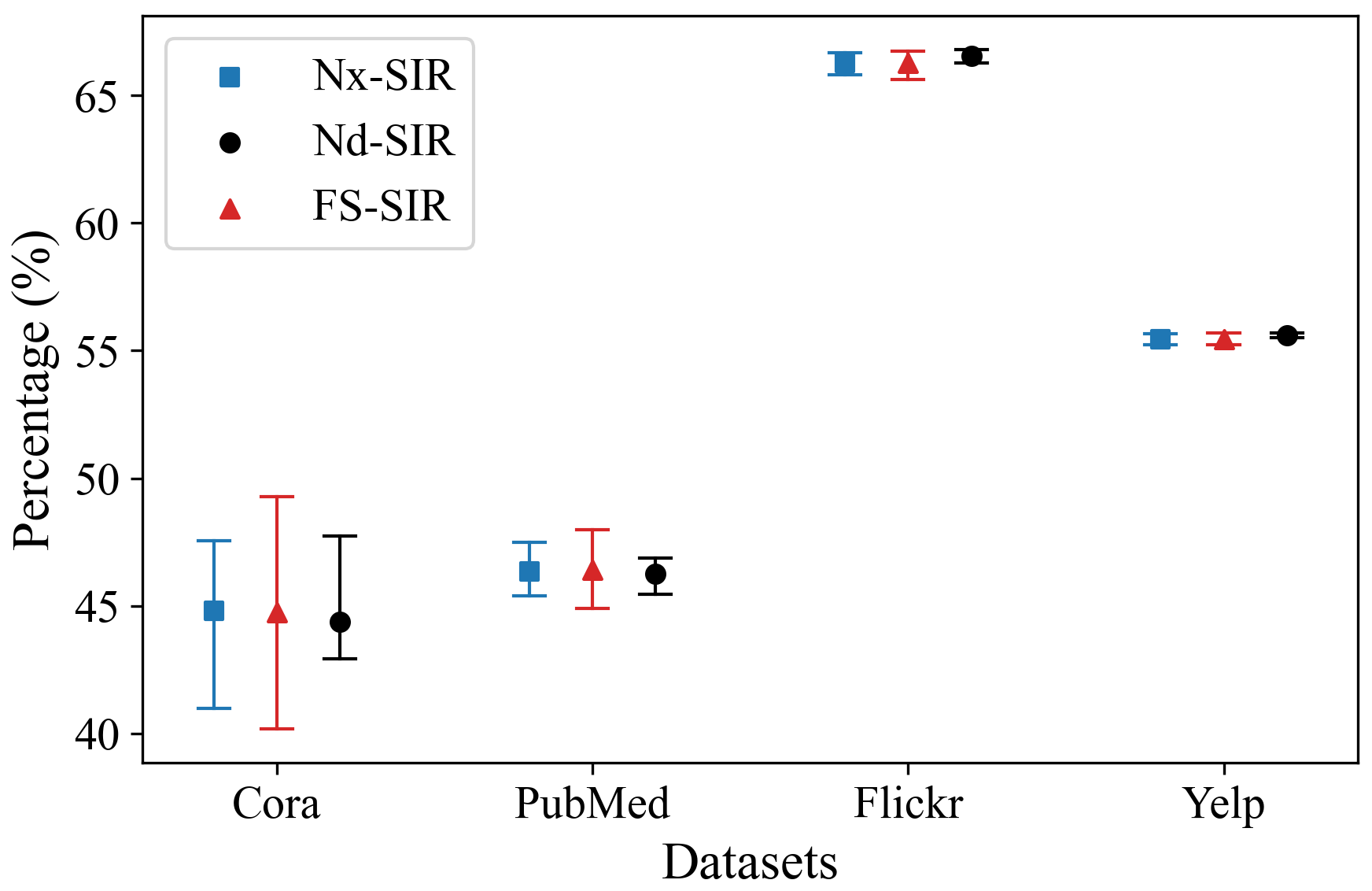}
    \end{minipage}
    \caption{Monte Carlo simulation results for three SIR implementations across datasets, illustrating FS\_GPlib’s consistency and accuracy.}
    \label{fig:SIR_3}
    \Description{}
\end{figure}

\subsection{Efficiency Benchmarking and Acceleration Analysis}

We evaluate the efficiency of FS\_GPlib through two complementary experiments.

\subsubsection{Implementation Runtime Comparison}\label{runtime}

To evaluate FS\_GPlib’s micro-level acceleration, we benchmark the runtime of 1{,}000 Monte Carlo simulations using several widely adopted implementations of the IC and SIR models: \textit{NetworkX} (Nx-*), \textit{NDlib} (Nd-*), CND-IC, GP-IC, EoN-SIR and GT-SIR. FS\_GPlib is tested in both CPU and GPU modes. Runtime results are presented in Tables~\ref{tab:IC_times} and~\ref{tab:SIR_times}, and Figures~\ref{fig:IC_4} and~\ref{fig:SIR_4} report relative speedups over the NDlib baseline.

\begin{table}
	\caption{Comparison of 1000 Monte Carlo simulations runtime performance (in seconds) among IC implementations on different datasets.}

\label{tab:IC_times}
\resizebox*{\linewidth}{!}{
	\begin{tabular}{ccccccc}
		\toprule
		Dataset& Nx-IC	&Nd-IC	&CND-IC& GP-IC&FS-IC(CPU)&	FS-IC(GPU)\\
		\midrule
		Cora & 17.79$^\dag$&	1233.64$^*$&\textbf{0.30}& 	1.01&5.00 &	3.01 \\
		PubMed & 167.24$^\dag$&	9223.87$^*	$& \textbf{2.43}& 3.46	&16.59 &	3.23  \\
		Flickr & 1689.79$^\dag$&	41669.81$^*	$& 26.30&49.66 &	50.45 &	\textbf{2.50}  \\
		Yelp &24644.47$^\dag$&	376977.50$^*	$& 	551.39& - &	1641.92& 	\textbf{10.75}  \\
		\bottomrule
	\end{tabular}
}
\begin{flushleft}
	\footnotesize

	$\dag$: This data is the time to perform 100 simulations multiplied by 10.\\
	$^*$: This data is the time to perform 10 simulations multiplied by 100.\\
\end{flushleft}
\end{table}

\begin{table}
\caption{Comparison of 1000 Monte Carlo simulations runtime performance (in seconds) among SIR implementations on different datasets.
}

\label{tab:SIR_times}
\resizebox*{\linewidth}{!}{
\begin{tabular}{ccccccc}
	\toprule
	Dataset &Nx-SIR&	Nd-SIR&	EoN-SIR&GT-SIR &FS-SIR(CPU)&	FS-SIR(GPU)\\
	\midrule
	Cora & 82.28$^\dag$&		314.11$^*	$& 17.15 &\textbf{6.02} &	60.17	&27.06 \\
	PubMed & 789.82$^\dag$&	2563.64$^*	$& 	177.00 &45.08	&177.60 &	\textbf{27.43 }   \\
	Flickr & 8172.21$^\dag$&		21055.04$^*	$& 	2240.73&168.55 &	898.33 &	\textbf{27.60}   \\
	Yelp &110627.07$^\dag$&	271245.34$^*	$& 	28765.07$^\dag$&  1546.66 &17975.58$^*$ &	\textbf{99.21}  \\
	\bottomrule
\end{tabular}
}
\begin{flushleft}
\footnotesize
$\dag$: This data is the time to perform 100 simulations multiplied by 10.\\
$^*$: This data is the time to perform 10 simulations multiplied by 100.\\
\end{flushleft}
\end{table}

\begin{figure}[ht]
	\centering
	\begin{minipage}{0.45\textwidth}
		\centering
		\includegraphics[width=\textwidth]{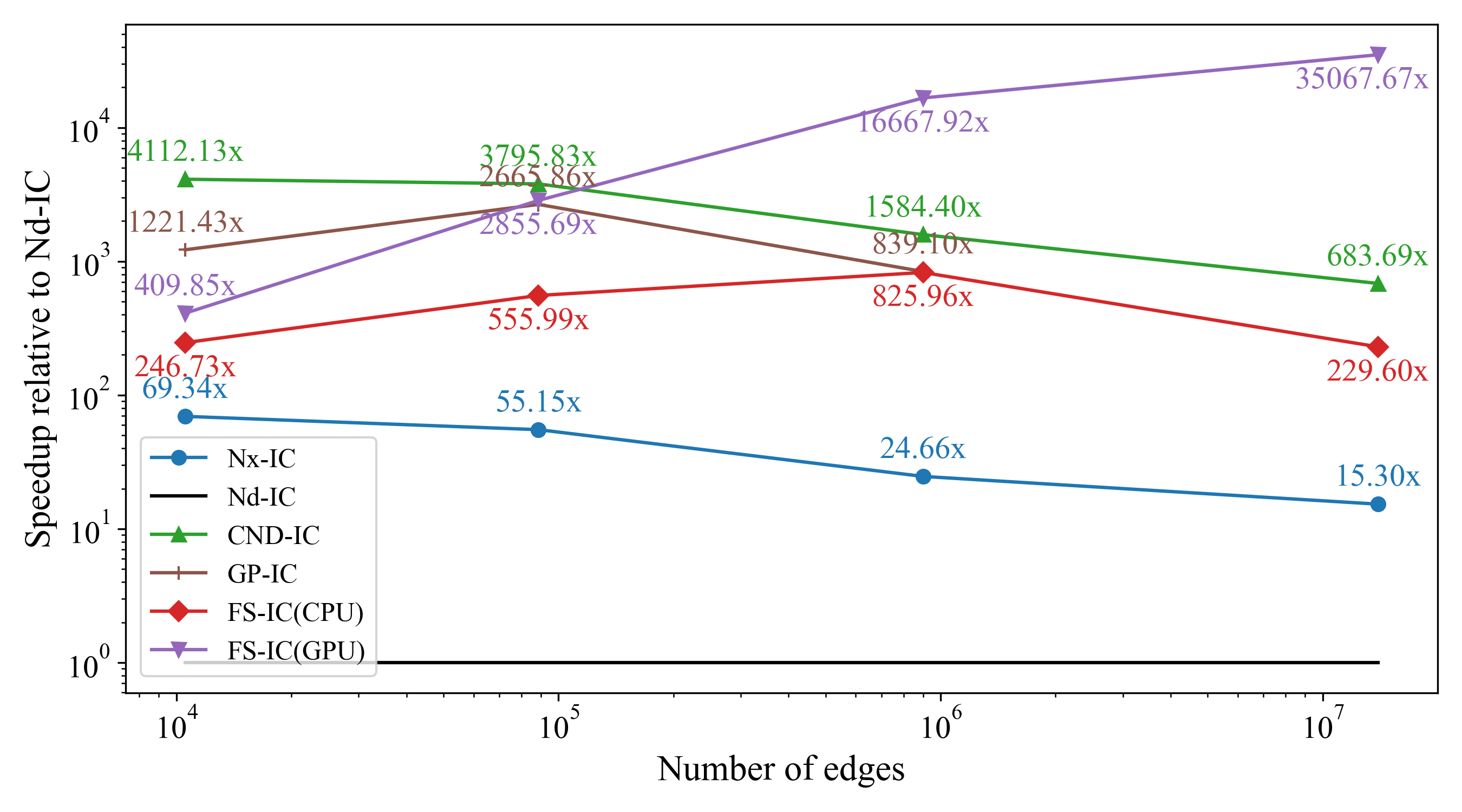}
	\end{minipage}
	\caption{Speedup of IC implementations over NDlib.}
	\label{fig:IC_4}
    \Description{}
\end{figure}

\begin{figure}[h]
	\centering
	\begin{minipage}{0.45\textwidth}
		\centering
		\includegraphics[width=\textwidth]{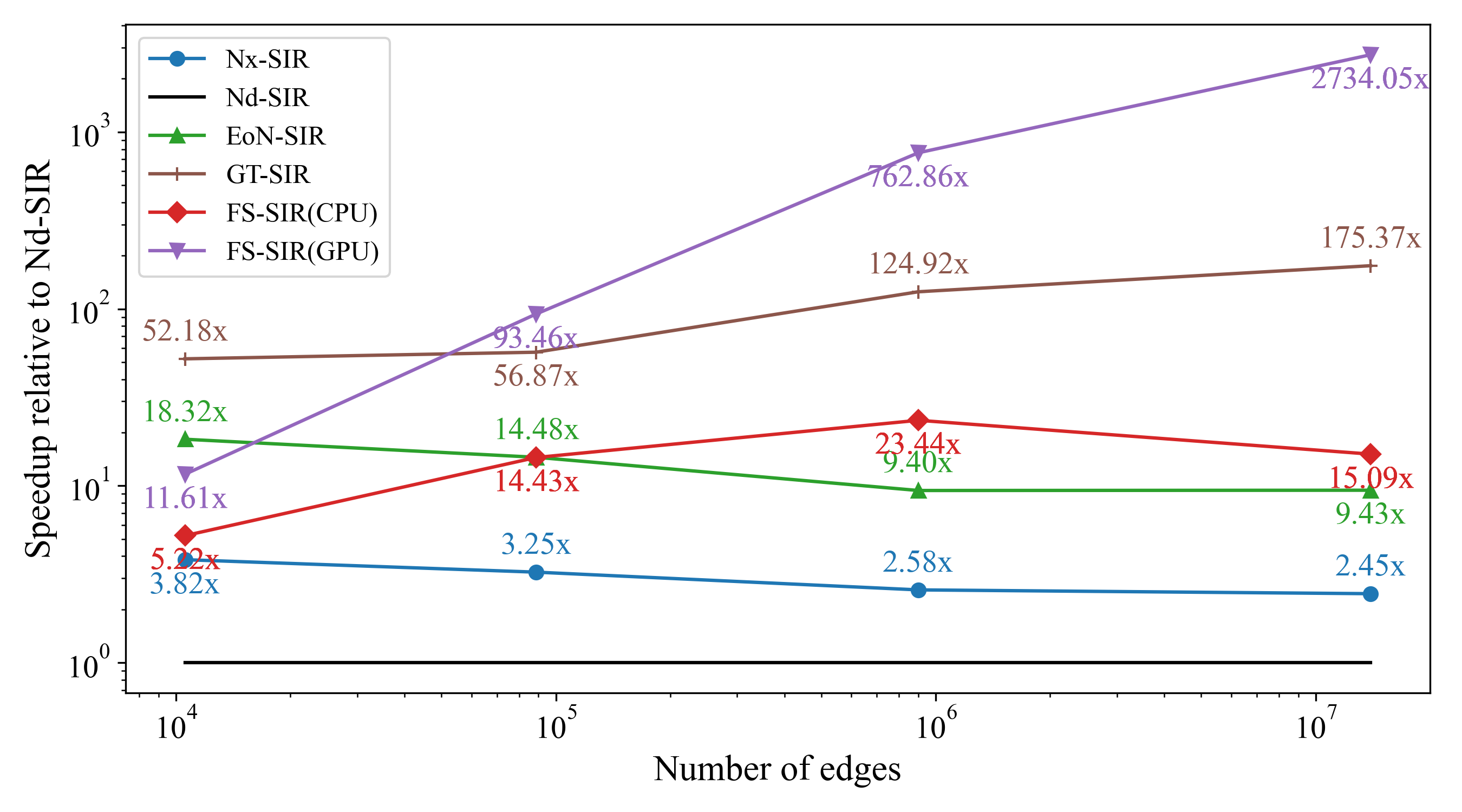}
	\end{minipage}
	\caption{Speedup of SIR implementations over NDlib.}
	\label{fig:SIR_4}
    \Description{}
\end{figure}

Among IC implementations, Nx-IC consistently outperforms Nd-IC but still exhibits high absolute runtimes. CND-IC performs best on small datasets; however, its efficiency gains diminish as graph size increases, indicating limited scalability. {\color{black}The performance of GP-IC is between CND-IC and FS-IC(CPU), but its processing of large-scale Yelp data exceeds memory, which is also the performance of limited scalability.} FS-IC(CPU) delivers competitive acceleration on small to medium graphs. In contrast, FS-IC(GPU) demonstrates robust scalability: its speedup relative to Nd-IC increases with dataset size, achieving over 35{,}000$\times$ acceleration on Yelp dataset—the shortest runtime among all implementations.

SIR simulation results follow a similar trend. Nx-SIR and EoN-SIR perform reasonably on smaller graphs but their efficiency drops as scale grows. FS-SIR(CPU) overtakes EoN-SIR on large datasets. GT-SIR runs the fastest on Cora data, but does not perform as well as FS-SIR(GPU) on larger datasets. Notably, the performance benefit of FS-SIR(GPU) significantly magnifies with graph size.

These results confirm FS\_GPlib’s strong micro-level acceleration capabilities, particularly its ability to scale efficiently on large networks by leveraging parallel hardware.

\subsubsection{Batch Size and Runtime Scaling}

We next evaluate FS\_GPlib’s macro-level acceleration by varying the batch size—i.e., the number of parallel Monte Carlo simulations executed per GPU pass. In theory, increasing batch size should linearly reduce total runtime due to amortized computation overhead. However, empirical results (Tables~\ref{tab:IC_batchsize} and \ref{tab:SIR_batchsize}) reveal diminishing returns beyond moderate sizes, especially on large-scale graphs where hardware bottlenecks begin to dominate.

\begin{corollary}[Optimal Batch Regime]
    There exists an upper bound on batch size $B_{\max}$ such that for any $B > B_{\max}$, the runtime of a single Monte Carlo simulation ceases to decrease and may even increase due to memory bandwidth saturation and GPU kernel scheduling overhead.
\end{corollary}

This observation provides practical guidance for selecting batch sizes in large-scale simulations: choosing $B$ close to $B_{\max}$ ensures maximal parallelism without incurring resource contention.

\begin{table}[h]
	%	\hspace{1.5cm}
	\begin{minipage}[t]{0.45\textwidth}
		\centering
		\caption{FS-IC Runtime of 1000 Monte Carlo simulations under varying batch sizes.}
		\label{tab:IC_batchsize}
		\begin{tabular}{ccccc}
			\toprule
			batch size&	1	&10	&100&	1000\\
			\midrule
			Cora&	3.01&	0.86& 	0.55& 	0.53\\ 
			PubMed&	3.23& 	0.91& 	0.61& 	0.61\\ 
			Flickr&	2.50& 	0.61& 	1.10& 	1.30\\ 
			Yelp&	10.75& 	12.32& 	12.84& 	-\\
			Pokec&	23.27& 	25.97& 	-&	-\\
			\bottomrule
		\end{tabular}
	\end{minipage}
\end{table}

\begin{table}[h]
	\centering
	\begin{minipage}[t]{0.45\textwidth}
		\centering
		\caption{FS-SIR Runtime of 1000 Monte Carlo simulations under varying batch sizes.}
		\label{tab:SIR_batchsize}
		\begin{tabular}{ccccc}
			\toprule
			batch size&	1	&10	&100&	1000\\
			\midrule
			Cora&	27.06& 	3.29 &	0.86& 	0.66   \\
			PubMed&	27.43 &	3.32 &	1.02& 	1.24  \\
			Flickr&	27.60 &	4.82 &	6.32 &	6.60   \\
			Yelp&	99.21 &	94.99 &	92.45 &	-  \\
			Pokec&	229.82 &	216.02& 	-	&-\\
			\bottomrule
		\end{tabular}
	\end{minipage}%
\end{table}

\subsection{Distributed Propagation on Web-Scale Graphs}

We evaluate FS\_GPlib's scalability on a Web-scale graph—Webbase, containing over 118 million nodes and 990 million edges. The experiment was conducted on a single machine with eight GPUs (see Section~\ref{subs:setup}), of which six were used for computation.

As shown in Table~\ref{tab:Distributed}, partitioning the graph required only 42 seconds. We then executed 1{,}000 Monte Carlo simulations across six processes, completing all runs in approximately $1.15 \times 10^4$ seconds—averaging just 11.5 seconds per full-scale simulation. The final infection ratio converged to 29.63\%.

\begin{table}[h!]
	\caption{Distributed SIR simulation result on the Webbase of 1000 Monte Carlo simulations.}
	\label{tab:Distributed}
	\resizebox*{\linewidth}{!}{
		\begin{tabular}{cccc}
			\toprule
			Dataset &Time of partition(s) &  Time of simulation(s)& Expected infection rate(\%) \\ %& Max of infected percentage(\%) & Min of infected percentage(\%) 
			\midrule
			Webbase &	42.08 	& 1.15E+04 &	29.63\\
			% & 29.68 & 29.58 \\
			\bottomrule
		\end{tabular}
	}
\end{table}
These empirical results corroborate Theorem~\ref{thm:constant_micro} and Corollary~\ref{cor:macro_batch}, confirming that FS\_GPlib achieves near-constant amortized runtime per simulation as graph size increases. This marks a breakthrough: FS\_GPlib reduces the perceived $O(m)$ complexity of traditional agent-based simulation to practical constant time per run in theory and enables near–real-time propagation modeling on billion-scale graphs empirically, unlocking practical research and engineering applications at unprecedented scale.

\section{Summary}\label{sec:summary}

%new
We introduced \textbf{FS\_GPlib}, a unified Python framework for fast, scalable propagation simulations on large graphs. By integrating a dual-accelerated design—message-passing-based agent modeling and batched execution, and a distributed propagation simulation method, FS\_GPlib bridges the longstanding gap between efficiency and scalability in propagation modeling.

The framework achieves significant improvements over existing libraries, delivering up to \textbf{35,000$\times$ speedup} and enabling simulation on graphs with over \textbf{100 million nodes}. Our design is supported by a new theoretical result: under GPU-accelerated batched Monte Carlo sampling, the amortized runtime per simulation is \textbf{asymptotically constant} under realistic batching and parallelism assumptions, even on web-scale networks. This formally explains FS\_GPlib’s near-constant wall-clock time in practice and validates its scalability guarantee. FS\_GPlib lays the foundation for high-throughput, high-fidelity, and theoretically grounded propagation modeling at scale.


\begin{acks}
    This work is supported by the National Natural Science Foundation of China (Grant Nos. T2293771, 62503447).

\end{acks}

\newpage
\bibliographystyle{ACM-Reference-Format}
\balance
\bibliography{main}

@String{Computing = "Computing" }

@String{Springer = "Springer-Verlag" }

@article{acemoglu2015systemic,
	title={Systemic risk and stability in financial networks},
	author={Acemoglu, Daron and Ozdaglar, Asuman and Tahbaz-Salehi, Alireza},
	journal={American Economic Review},
	volume={105},
	number={2},
	pages={564--608},
	year={2015},
    month={},
    doi={10.1257/aer.20130456}
}

@inproceedings{barrett2008episimdemics,
	title={Episimdemics: an efficient algorithm for simulating the spread of infectious disease over large realistic social networks},
	author={Barrett, Christopher L and Bisset, Keith R and Eubank, Stephen G and Feng, Xizhou and Marathe, Madhav V},
	booktitle={SC'08: Proceedings of the 2008 ACM/IEEE Conference on Supercomputing},
    year={2008},
    volume={},
    number={},
    pages={1--12},
	organization={IEEE},
    doi={10.1109/SC.2008.5214892}
}

@inproceedings{boldi2004webgraph,
	title={The webgraph framework I: compression techniques},
	author={Boldi, Paolo and Vigna, Sebastiano},
	booktitle={Proceedings of the 13th international conference on World Wide Web},
	pages={595--602},
	year={2004},
    publisher = {Association for Computing Machinery},
    address = {New York, NY, USA},
    doi = {10.1145/988672.988752},
    series = {WWW '04}
}

@article{daley1964epidemics,
	title={Epidemics and rumours},
	author={Daley, Daryl J and Kendall, David G},
	journal={Nature},
	volume={204},
	number={4963},
	pages={1118--1118},
	year={1964},
	publisher={Nature Publishing Group UK London},
    doi={10.1038/2041118a0}
}

@article{devarapalli2024estimating,
	title={Estimating rumor source in social networks using incomplete observer information},
	author={Devarapalli, Ravi Kishore and Biswas, Anupam},
	journal={Expert Systems with Applications},
	volume={249},
	pages={123499},
	year={2024},
	publisher={Elsevier},
    doi = {10.1016/j.eswa.2024.123499}
}

@article{ghosh2022dynamics,
	title={Dynamics and control of delayed rumor propagation through social networks},
	author={Ghosh, Moumita and Das, Samhita and Das, Pritha},
	journal={Journal of Applied Mathematics and Computing},
	pages={1--30},
	year={2022},
	publisher={Springer},
    doi={10.1007/s12190-021-01643-5}
}

@article{gillespie2001approximate,
	title={Approximate accelerated stochastic simulation of chemically reacting systems},
	author={Gillespie, Daniel T},
	journal={The Journal of Chemical Physics},
	volume={115},
	number={4},
	pages={1716--1733},
	year={2001},
	publisher={American Institute of Physics},
    doi={10.1063/1.1378322}
}

@misc{apacheGiraph,
  title = {Giraph - {Welcome} {To} {Apache} {Giraph!}},
  author = {{Apache Software Foundation}},
  year = {2020},
  month = aug,
  url = {https://giraph.apache.org/},
  note = {Accessed: 2026-01-24.},
  howpublished = {Apache Giraph Official Website}
}

@article{goldenberg2001talk,
	title={Talk of the network: A complex systems look at the underlying process of word-of-mouth},
	author={Goldenberg, Jacob and Libai, Barak and Muller, Eitan},
	journal={Marketing letters},
	volume={12},
	pages={211--223},
	year={2001},
	publisher={Springer},
    doi={10.1023/A:1011122126881}
}

@article{goldenberg2001using,
	title={Using complex systems analysis to advance marketing theory development: Modeling heterogeneity effects on new product growth through stochastic cellular automata},
	author={Goldenberg, Jacob and Libai, Barak and Muller, Eitan},
	journal={Academy of Marketing Science Review},
	volume={9},
	number={3},
	pages={1--18},
	year={2001},
	publisher={Citeseer},

}

@inproceedings{gonzalez2012powergraph,
    title = {{PowerGraph}: Distributed {Graph-Parallel} Computation on Natural Graphs},
	author={Gonzalez, Joseph E and Low, Yucheng and Gu, Haijie and Bickson, Danny and Guestrin, Carlos},
	booktitle={10th USENIX symposium on operating systems design and implementation (OSDI 12)},
	pages={17--30},
	year={2012},
	address = {Hollywood, CA},
    publisher = {USENIX Association}
}

@inproceedings{gonzalez2014graphx,
	title={{GraphX}: Graph processing in a distributed dataflow framework},
	author={Gonzalez, Joseph E and Xin, Reynold S and Dave, Ankur and Crankshaw, Daniel and Franklin, Michael J and Stoica, Ion},
	booktitle={11th USENIX symposium on operating systems design and implementation (OSDI 14)},
	pages={599--613},
	year={2014}
}

@article{granovetter1978threshold,
	title={Threshold models of collective behavior},
	author={Granovetter, Mark},
	journal={American journal of sociology},
	volume={83},
	number={6},
	pages={1420--1443},
	year={1978},
	publisher={University of Chicago Press},
    doi={10.1086/226707}
}

@article{haldane2011systemic,
	title={Systemic risk in banking ecosystems},
	author={Haldane, Andrew G and May, Robert M},
	journal={Nature},
	volume={469},
	number={7330},
	pages={351--355},
	year={2011},
	publisher={Nature Publishing Group UK London},
    doi={10.1038/nature09659}
}

@article{jain2023optimal,
	title={Optimal control of rumor spreading model on homogeneous social network with consideration of influence delay of thinkers},
	author={Jain, Ankur and Dhar, Joydip and Gupta, Vijay K},
	journal={Differential Equations and Dynamical Systems},
	volume={31},
	number={1},
	pages={113--134},
	year={2023},
	publisher={Springer},
    doi={10.1007/s12591-019-00484-w}
}

@article{jenness2018epimodel,
	title={EpiModel: an R package for mathematical modeling of infectious disease over networks},
	author={Jenness, Samuel M and Goodreau, Steven M and Morris, Martina},
	journal={Journal of Statistical Software},
	volume={84},
	pages={1--47},
	year={2018},
    doi={10.18637/jss.v084.i08}
}

@article{karypis1998fast,
	title={A fast and high quality multilevel scheme for partitioning irregular graphs},
	author={Karypis, George and Kumar, Vipin},
	journal={SIAM Journal on Scientific Computing},
	volume={20},
	number={1},
	pages={359--392},
	year={1998},
	publisher={SIAM},
    doi={10.1137/S1064827595287997}
}

@inproceedings{kempe2003maximizing,
	title={Maximizing the spread of influence through a social network},
	author={Kempe, David and Kleinberg, Jon and Tardos, {\'E}va},
	booktitle={Proceedings of the ninth ACM SIGKDD international conference on Knowledge discovery and data mining},
	pages={137--146},
	year={2003},
    location = {Washington, D.C.},
    series = {KDD '03},
    doi = {10.1145/956750.956769}
}

@article{kermack1927contribution,
	title={A contribution to the mathematical theory of epidemics},
	author={Kermack, William Ogilvy and McKendrick, Anderson G},
	journal={Proceedings of the Royal Society of London. Series A, Containing Papers of a Mathematical and Physical Character},
	volume={115},
	number={772},
	pages={700--721},
	year={1927},
	publisher={The Royal Society London},
    doi={10.1098/rspa.1927.0118}
}

@book{kiss2017mathematics,
  title     = {Mathematics of Epidemics on Networks},
  author    = {Kiss, Istv{\'a}n Z and Miller, Joel C and Simon, P{\'e}ter L},
  publisher = {Springer},
  address   = {Cham},
  year      = {2017},
  doi       = {10.1007/978-3-319-50806-1}
}

@inproceedings{koster2022gpu,
	title={GPU-Accelerated Simulation Ensembles of Stochastic Reaction Networks},
	author={K{\"o}ster, Till and Herrmann, Leon and Andelfinger, Philipp and Uhrmacher, Adelinde},
	booktitle={2022 Winter Simulation Conference (WSC)},
    volume={},
    number={},
	pages={2570--2581},
	year={2022},
	organization={IEEE},
    doi={10.1109/WSC57314.2022.10015448}
}

@article{li2022influence,
	title={An influence maximization method based on crowd emotion under an emotion-based attribute social network},
	author={Li, Weimin and Li, Yaqiong and Liu, Wei and Wang, Can},
	journal={Information Processing \& Management},
	volume={59},
	number={2},
	pages={102818},
	year={2022},
	publisher={Elsevier},
    doi={10.1016/j.ipm.2021.102818},
}

@inproceedings{ling2023deep,
	title={Deep graph representation learning and optimization for influence maximization},
	author={Ling, Chen and Jiang, Junji and Wang, Junxiang and Thai, My T and Xue, Renhao and Song, James and Qiu, Meikang and Zhao, Liang},
	booktitle={International conference on machine learning},
	pages={21350--21361},
	year={2023},
	organization={PMLR}
}

@inproceedings{malewicz2010pregel,
	title={Pregel: a system for large-scale graph processing},
	author={Malewicz, Grzegorz and Austern, Matthew H and Bik, Aart JC and Dehnert, James C and Horn, Ilan and Leiser, Naty and Czajkowski, Grzegorz},
	booktitle={Proceedings of the 2010 ACM SIGMOD International Conference on Management of data},
	pages={135--146},
	year={2010},
    doi={10.1145/1807167.1807184},
    location = {Indianapolis, Indiana, USA},
    series = {SIGMOD '10}
}

@article{miller2020eon,
	title={Eon (epidemics on networks): a fast, flexible python package for simulation, analytic approximation, and analysis of epidemics on networks},
	author={Miller, Joel C and Ting, Tony},
	journal={arXiv preprint arXiv:2001.02436},
	year={2020},
    doi={10.48550/arXiv.2001.02436}
}

@article{pastor2015epidemic,
	title={Epidemic processes in complex networks},
	author={Pastor-Satorras, Romualdo and Castellano, Claudio and Van Mieghem, Piet and Vespignani, Alessandro},
	journal={Reviews of Modern Physics},
	volume={87},
	number={3},
	pages={925--979},
    year={2015},
    publisher = {American Physical Society},
    doi = {10.1103/RevModPhys.87.925}
}

@article{peixoto_graph-tool_2014,
	title = {The graph-tool python library},
	url = {http://figshare.com/articles/graph_tool/1164194},
	doi = {10.6084/m9.figshare.1164194},
	urldate = {2014-09-10},
	journal = {figshare},
	author = {Peixoto, Tiago P.},
	year = {2014}
}

@article{rainer2002opinion,
	title={Opinion dynamics and bounded confidence: models, analysis and simulation},
	author={Hegselmann, Rainer and Krause, Ulrich},
	year={2002},
    journal = {Journal of Artificial Societies and Social Simulation},
    volume = {5},
    number = {3}
}

@article{robson2024cynetdiff,
	title={CyNetDiff: A Python Library for Accelerated Implementation of Network Diffusion Models},
	author={Robson, Eliot W and Reddy, Dhemath and Umrawal, Abhishek K},
	journal={Proceedings of the VLDB Endowment},
	volume={17},
	number={12},
	pages={4409--4412},
	year={2024},
	publisher={Very Large Data Base Endowment Inc.},
    doi={10.14778/3685800.3685887},
}

@inproceedings{rossetti2018ndlib,
	title={NDlib: a python library to model and analyze diffusion processes over complex networks},
	author={Rossetti, Giulio and Milli, Letizia and Rinzivillo, Salvatore},
	booktitle={Companion Proceedings of the The Web Conference 2018},
	pages={183--186},
	year={2018},
    doi={10.1145/3184558.3186974},
    location = {Lyon, France},
    series = {WWW '18}
}

@article{samaei2025fastgemf,
	title={FastGEMF: Scalable High-Speed Simulation of Stochastic Spreading Processes Over Complex Multilayer Networks},
	author={Samaei, Mohammad Hossein and Sahneh, Faryad Darabi and Scoglio, Caterina},
	journal={IEEE Access},
	year={2025},
    volume={13},
    number={},
    pages={27112--27125},
	publisher={IEEE},
    doi={10.1109/ACCESS.2025.3539345}
}

@article{su2025gpic,
	title={GPIC: a GPU-based parallel independent cascade algorithm in complex networks},
	author={Su, Chang and Na, Xu and Zhou, Fang and L{\"u}, Linyuan},
	journal={Chinese Physics B},
	year={2025},
    doi = {10.1088/1674-1056/adb67c},
    volume = {34},
    number = {3},
    pages = {030204}
}

@inproceedings{sun2020influence,
	title={Influence maximization with spontaneous user adoption},
	author={Sun, Lichao and Chen, Albert and Yu, Philip S and Chen, Wei},
	booktitle={Proceedings of the 13th International Conference on Web Search and Data Mining},
	pages={573--581},
	year={2020},
    doi = {10.1145/3336191.3371791},
    numpages = {9},
    location = {Houston, TX, USA},
    series = {WSDM '20}
}

@inproceedings{takac2012data,
	title={Data analysis in public social networks},
	author={Takac, Lubos and Zabovsky, Michal},
	booktitle={International scientific conference and international workshop present day trends of innovations},
	volume={1},
	number={6},
	year={2012}
}

@inproceedings{yang2016revisiting,
	title={Revisiting semi-supervised learning with graph embeddings},
	author={Yang, Zhilin and Cohen, William and Salakhudinov, Ruslan},
	booktitle={International conference on machine learning},
	pages={40--48},
	year={2016},
    volume={48},
	organization={PMLR},

}

@article{youssef2011individual,
	title={An individual-based approach to SIR epidemics in contact networks},
	author={Youssef, Mina and Scoglio, Caterina},
	journal={Journal of Theoretical Biology},
	volume={283},
	number={1},
	pages={136--144},
	year={2011},
	publisher={Elsevier},
    doi={10.1016/j.jtbi.2011.05.029}
}

@article{zeng2019graphsaint,
	title={Graphsaint: Graph sampling based inductive learning method},
	author={Zeng, Hanqing and Zhou, Hongkuan and Srivastava, Ajitesh and Kannan, Rajgopal and Prasanna, Viktor},
	journal={arXiv preprint arXiv:1907.04931},
	year={2019},
    doi={10.48550/arXiv.1907.04931}
}

\newpage

\appendix
%\renewcommand{\thefigure}{A\arabic{figure}} 
%\renewcommand{\thetable}{A\arabic{table}}  
%\setcounter{figure}{0}  
%\setcounter{table}{0} 
%\renewcommand{\thesubsection}{A\arabic{subsection}} 
%
%\section*{Appendix}\label{sec:appendix}

\section{Examples of Accelerated Model Implementations}
\label{SIRHK}
We illustrate our framework using two classical models: the discrete-state SIR model and the continuous-state HK model.

\subsection{Micro-accelerated implementation of SIR Model}

The SIR Model assumes that infection spreads only through links between neighboring nodes in a graph $G=(V,E)$. Each node is in one of three states: $S$, $I$, or $R$. A susceptible node $i$ is infected by its infected neighbors $j \in N(i)$ with rate $\beta$, while infected nodes recover with rate $\gamma$:
\begin{equation}
	\frac{dS_i}{dt} = -\beta \sum_{j\in N(i)} S_i I_j ,\quad
	\frac{dI_i}{dt} = \beta \sum_{j\in N(i)} S_i I_j - \gamma I_i ,\quad
	\frac{dR_i}{dt} = \gamma I_i .
	\label{eq:network_SIR}
\end{equation}

According to Figure \ref{fig:SIR}, node transitions follow two rules: 1) if a \textit{S} state node has \textit{I} state neighbors, each \textit{I} state neighbor transmits the infection to the \textit{S} state node with probability $\beta$; 2) the \textit{I} state node is recovered to the \textit{R} state with probability $\gamma$.

\begin{figure}[h]
	\centering
	\includegraphics[width=0.7\linewidth]{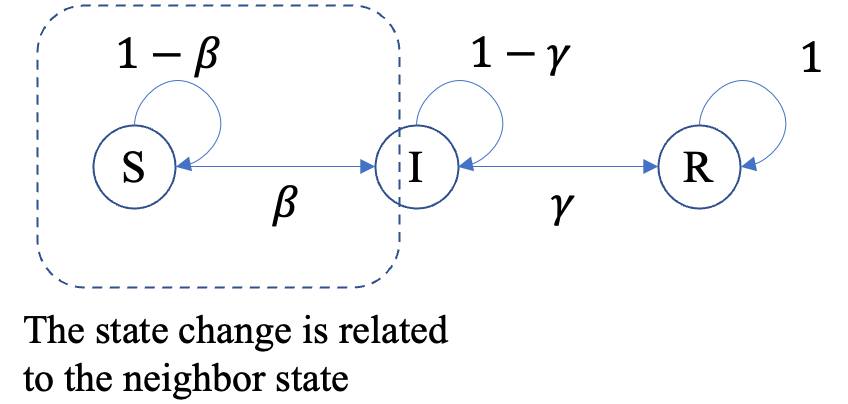}
	\caption{State transition diagram of SIR model.}
	\Description{}
	\label{fig:SIR}
\end{figure}

\begin{definition}[Micro-accelerated implementation of SIR Model]

	We represent node states using two Boolean indicator vectors $h, r \in \{0,1\}^N$, 
	where $h_i=1$ denotes \emph{infected}, $r_i=1$ denotes \emph{recovered}, and $(h_i,r_i)=(0,0)$ denotes \emph{susceptible}. 
	The update of the system at step $k$ is decomposed into three stages:

Each infected neighbor $j$ of node $i$ transmits a log-probability contribution

	\begin{equation}
	m_{ji}^{(k)} = h_j^{(k-1)} \cdot \log(1-\beta)
	\end{equation}
	
Node $i$ collects contributions from all neighbors $N(i)$ to compute its infection probability

	\begin{equation}
	m_i^{(k)} = 1 - \exp\!\left( \sum_{j \in N(i)} m_{ji}^{(k)} \right)
	\end{equation}

The indicator variables are updated with independent uniform random variables $U_i^{\mathrm{inf}}, U_i^{\mathrm{rec}} \sim \mathrm{Uniform}(0,1)$
\begin{equation}
\begin{aligned}
	h_i^{(k)} &=
	\begin{cases}
		1, & \text{if } (h_i^{(k-1)}=0 \land r_i^{(k-1)}=0) \land (U_i^{\mathrm{inf}} < m_i^{(k)}), \\[4pt]
		1, & \text{if } h_i^{(k-1)}=1 \land (U_i^{\mathrm{rec}} \ge \delta), \\[4pt]
		0, & \text{otherwise},
	\end{cases} \\[6pt]
	r_i^{(k)} &=
	\begin{cases}
		1, & \text{if } h_i^{(k-1)}=1 \land (U_i^{\mathrm{rec}} < \delta), \\[4pt]
		r_i^{(k-1)}, & \text{otherwise}.
	\end{cases}
\end{aligned}
\end{equation}

\end{definition}

After simulation, Boolean features are converted into an $N$-length state vector with entries 0, 1, or 2, denoting \textit{S}, \textit{I}, and \textit{R}.

\subsection{Micro-accelerated implementation of HK model}

The HK model \cite{rainer2002opinion} describes continuous opinion dynamics. Each node holds an initial opinion $h \in [-1,1]$ and updates it in discrete rounds. At each step: 1) $\epsilon$ neighbor identification: find the set $\Gamma_\epsilon$ of each node, $d_{i,j}=\left|h_i-h_j\right|\le\epsilon$, $j\in\Gamma_\epsilon$; 2) update the opinion value: The opinion of node $i$ at the next time step is updated as follows:

\begin{equation}
	h_i^{\left(k\right)}=\frac{\sum_{j\in\Gamma_\epsilon}{h_j^{\left(k-1\right)}}}{\#\Gamma_\epsilon},
	\label{eq:HK}
\end{equation}
where $\#\Gamma_\epsilon$ denotes the number of $\epsilon$ neighbors.

\begin{figure}[h]
	\centering
	\includegraphics[width=0.6\linewidth]{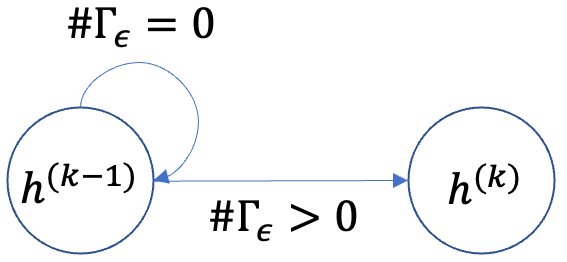}
	\caption{State transition diagram of HK model.}
	\Description{}
	\label{fig:HK}
\end{figure}

\begin{definition}[Micro-accelerated implementation of HK Model]
	Unlike the SIR model with three discrete states, the HK model uses a continuous opinion space. To avoid infinitely many states, we consider only two cases (Figure \ref{fig:HK}): if $\#\Gamma_\epsilon=0$, the opinion remains unchanged; otherwise, it is updated to the $\epsilon$-neighbors’ average.
	
	For each neighbor $j \in N(i)$, generate two message terms:
\begin{equation}
\begin{aligned}
	m_{ji\_c}^{(k)} &=\mathbf{1}\!\left( |h_i^{(k-1)} - h_j^{(k-1)}| < \varepsilon \right) \cdot h_j^{(k-1)}, \\[6pt]
	m_{ji\_v}^{(k)} &= \mathbf{1}\!\left( |h_i^{(k-1)} - h_j^{(k-1)}| < \varepsilon \right).
\end{aligned}
\end{equation}
where $\mathbf{1}(\cdot)$ is the indicator function.

Node $i$ aggregates received messages by averaging opinions of $\varepsilon$-neighbors:
\begin{equation}
\begin{aligned}
	m_i^{(k)} &= 
	\begin{cases}
		\dfrac{\sum\limits_{j \in N(i)} m_{ji\_c}^{(k)}}{\sum\limits_{j \in N(i)} m_{ji\_v}^{(k)}}, 
		& \text{if } \sum\limits_{j \in N(i)} m_{ji\_v}^{(k)} > 0, \\[8pt]
	\mathrm{NaN}, & \text{otherwise}.
	\end{cases} \\[10pt]
\end{aligned}
\end{equation}

Finally, the opinion of node $i$ is updated as
\begin{equation}
	\begin{aligned}
	h_i^{(k)} &= 
	\begin{cases}
		m_i^{(k)}, & \text{if } m_i^{(k)} \neq \mathrm{NaN}, \\[6pt]
		h_i^{(k-1)}, &  \text{otherwise}.
	\end{cases}
\end{aligned}
\end{equation}

\end{definition}
%This becomes the node's feature at the next iteration. 
The node feature $h$, which is the node's opinion value, is returned to the user after all iterations.

\section{Proofs}
\label{2proof}
\begin{proof}[Proof of Theorem~\ref{thm:constant_micro}]
	For each propagation step, the computational complexity is dominated by edge-level operations
	%\st{ (message passing and aggregation)}
	, requiring $O(m)$ work sequentially. When $p = \Theta(m)$ parallel processors can simultaneously accommodate all $m$ edges, the per-processor workload becomes $O(1)$. Under ideal parallel execution with negligible overhead, all processors complete their assigned work in parallel within constant time, achieving:
	\[
	T_{\text{step}} = O\left(\frac{m}{p}\right) = O(1).
	\]
	Thus, the runtime per update step is constant.
\end{proof}

\begin{proof}[Proof of Corollary~\ref{cor:macro_batch}]
	Running one simulation individually requires $O(1)$ time for $K=O(1)$ steps. For $S$ serial Monte Carlo simulations, it takes $O(S)$ time.
	With batch-parallel execution, $B$ simulations share tensor operations, resulting in total runtime $O(1)$ for the entire batch:
	\[
	T_{\text{sim}} = \frac{O(S)}{B} = O\left(\frac{S}{S}\right) = O(1).
	\]

	Hence, under batch-parallel settings, the runtime of all Monte Carlo simulations is constant.
\end{proof}

\section{Examples of Use}\label{example}
Listing \ref{lst:example} demonstrates usage: loading Cora dataset, selecting seeds via degree centrality, specifying SIR model parameters, performing 10,000 Monte Carlo simulations with batch processing, and computing the final propagation range.

\begin{figure}[ht]
%	\centering
	\begin{lstlisting}[style=mypython, caption=Python example: running Monte Carlo simulations using the SIR model with FS\_GPlib, label={lst:example}]
import torch
from torch_geometric.datasets import Planetoid
from torch_geometric.utils import degree
from FS_GPlib.Epidemics.SIRModel import SIRModel

def Degree_centrality_seeds(data, num_seeds):
	target = data.edge_index[1]
	D = degree(target, data.num_nodes)
	_, index = torch.sort(D, descending=True)
	seeds = index[:num_seeds]
	return seeds.tolist()

dataset = Planetoid(root='./dataset', name='Cora')
data = dataset[0]
num_seeds = int(data.num_nodes*0.1)
seeds = Degree_centrality_seeds(data, num_seeds)

i_beta=0.01 # infection
r_lambda=0.005 # recovery
d = 1 # device
MC = 10000 # simulation times
it = 100 # iteration times
bs = 2000 # batch size
model = SIRModel(data, seeds, i_beta, r_lambda, d)
finals = model.run_epochs(MC, it, bs)

count = (finals>0).sum().item()/MC
print(f'Final spread range is {count}')                                                                        
	\end{lstlisting}
    \Description{}
\end{figure}

\section{Performance of Different Devices}
\label{app:devices}

Considering that researchers work in diverse environments, this section assesses the performance of the algorithm across different devices. We selected four platforms: Server-Linux-CPU (SLC), server-Linux-GPU (SLG), PC-Linux-GPU (PLG), and PC-MacOS-CPU (PMC). Note that PLG setup uses WSL2 to run the Linux kernel within Windows. Table \ref{tab:Equipment} lists the hardware specifications and operating systems for each platform. Remaining software configurations are detailed in Section~\ref{subs:setup}.

\begin{table*}[ht]
    \caption{Equipment information of four devices.}
    \label{tab:Equipment}
    \resizebox*{\linewidth}{!}{
        \begin{tabular}{ccccc}
            \toprule
            &SLC&	SLG&	PLG	&PMC\\
            \midrule
            CPU &	Intel(R) Xeon(R) Platinum 8260 CPU @ 2.40GHz $\times$ 4	&Intel(R) Xeon(R) Platinum 8358 CPU @ 2.60 GHz $\times$ 2&	13th Gen Intel(R) Core(TM) i7-13700F&	Apple M1 Pro\\ 
            Core/Thread	& 24 Core/48 Thread&	32 Core/64 Thread&	12 Core/24 Thread	&8 Core \\ 
            Main frequency&	2.40 GHz&	2.60 GHz &	2.11 GHz&	-\\ 
            Memory & 1.5 TiB&	503 GiB&	7.6 GiB&	16 GB\\
            GPU & -	& NVIDIA GeForce RTX 4090 $\times$ 8 &NVIDIA GeForce RTX 4060 Ti $\times$ 1 &	-\\
            GPU video memory&	-&	24.5 GiB $\times$ 8&	8.0 GiB&	-\\
            Operating system/kernel	&Rocky Linux 8.5 &	Rocky Linux 8.5 &Ubuntu 22.04.5 LTS(WSL20	&macOS 15.2\\
            Other	& - &	CUDA\_VISIBLE\_DEVICES=0	& & \\	
            
            \bottomrule
        \end{tabular}
    }
\end{table*}

\begin{table*}[h!]
    \centering
    \begin{minipage}[t]{0.45\textwidth}
        \centering
        \caption{FS-IC Runtime of 1000 Monte Carlo simulations under varying devices.
        }
        %Running time of the IC model for four devices.}
    \label{tab:IC_device}
    \begin{tabular}{lcccc}
        \toprule
        device	&\  SLC&\ 	PLG	&\ \ SLG&PMC \\
        \midrule
        Cora&	\ 2.69&	\ 3.30&\ \ 3.01&	\textbf{1.41}  \\
        PubMed	&\ 10.15&\ 	3.80&\ \ \textbf{3.23}&	5.24  \\
        Flickr&	\ 31.16&	\ \textbf{2.50}&\ \ \textbf{2.50}&	23.75 \\
        Yelp&	\ 1271.11$^\dag$&\ 	33.34&\ \ \textbf{10.75}&	597.65 \\
        Pokec&	\ 2863.42$^\dag$&	\ 71.90&\ \ \textbf{23.27}&	1487.83$^\dag$\\
        \bottomrule
        \multicolumn{5}{l}{
            \footnotesize        
            $\dag$: This data is the time to perform 100 simulations multiplied by 10.
        }
    \end{tabular}

\end{minipage}%
\hspace{1.5cm}
\begin{minipage}[t]{0.45\textwidth}
    \centering
    \caption{FS-SIR Runtime of 1000 Monte Carlo simulations under varying devices.
        %Running time of the SIR model for four devices.
    }
    \label{tab:SIR_device}
    \begin{tabular}{lcccc}
        \toprule
        device	&SLC&	PLG	&SLG&PMC \\
        \midrule
        Cora &	30.73 	&26.49 &	27.06 &	\textbf{12.78} \\
        PubMed	&109.28& 	25.80& 	\textbf{27.43 }&	53.42 \\
        Flickr&	574.71& 	28.01& 	\textbf{27.60} &	431.03 \\
        Yelp&	13776.55$^*$&	367.62& \textbf{99.21} &	6888.00$^\dag$ \\
        Pokec&	31129.42$^*$&	868.66$^\dag$ &	\textbf{229.83 }&	16812.12$^*$\\
        \bottomrule
        \multicolumn{5}{l}{
            \footnotesize
            $\dag$: This data is the time to perform 100 simulations multiplied by 10.
            
        }\\
        \multicolumn{5}{l}{
            \footnotesize
            $^*$: This data is the time to perform 10 simulations multiplied by 100.
        }
    \end{tabular}

\end{minipage}
\end{table*}

Table~\ref{tab:IC_device} and Table~\ref{tab:SIR_device} present the runtime of IC and SIR models on four platforms using the same parameters as in Section~\ref{runtime} (batch size = 1). Each runtime is the average of three independent runs, and between runs we ensured device temperatures returned to normal to avoid thermal‐throttling artifacts. 
SLG dominates overall due to superior parallel acceleration, and PLG follows closely. However, PMC outperforms both GPUs on small-scale data with minimal batch sizes where GPU utilization is inefficient. The overall running time of SLC is relatively long, especially for large-scale data.

This experiment evaluated diverse computing environments (servers vs. PCs, GPU vs. CPU, different OS).  The algorithms ran successfully across all configurations with strong performance, offering practical guidance for device selection.

\section{Accuracy Validation of FS\_GPlib}\label{app:accuracy_other_models}

We validated FS\_GPlib on Cora and PubMed, with results (Table~\ref{tab:result}) consistent across NDlib, NetworkX, and our implementation. Algorithm settings are detailed in the code.

Due to computational constraints, not all models completed 1000 Monte Carlo simulations. Specifically, the Threshold, HK, and WHK models were executed only once. Simulation results show the expected counts of activated/infected nodes. For the HK and WHK models, nodes with a viewpoint value exceeding 0.5 were designated as activated/infected nodes. Notably, results of SWIR represent the expected sum of weak, infected, and recovered nodes. Implementation differences between some models and their NDlib counterparts precluded their inclusion in the comparison.

These computational results empirically establish FS\_GPlib's accuracy at the experimental level, thereby providing foundational validation for the efficiency and scalability improvements demonstrated in the main text.
\begin{table*}[htbp]
    \centering
    \caption{Simulation results for the three implementations.}
    \label{tab:result}
    \begin{tabular}{lp{0.09\textwidth}p{0.09\textwidth}p{0.09\textwidth}p{0.09\textwidth}p{0.09\textwidth}p{0.09\textwidth}}
        \toprule
        & \multicolumn{3}{c}{\textbf{Cora}} & \multicolumn{3}{c}{\textbf{PubMed}} \\
        \cmidrule(lr){2-4}\cmidrule(lr){5-7}
        \textbf{Model (parameters)} 
        & \textbf{NDlib} 
        & \textbf{Networkx} 
        & \textbf{FS\_GPlib} 
        & \textbf{NDlib} 
        & \textbf{Networkx} 
        & \textbf{FS\_GPlib} \\
        \midrule
        SIR (0.01,0.005) 
        & 1212.70
        & 1210.51 
        & 1211.64 
        & 9153.00
        & 9158.47 
        & 9151.78 \\
        IC (0.5) 
        & 1842.00$^\dag$
        & 1844.87$^\dag$
        & 1844.791
        & 13319.50
        & 13319.32 
        & 13318.44 \\
        SI (0.01) 
        & 1857.50$^*$
        & 1866.22$\dag$
        & 1868.57
        & 14235.8$^*$
        & 14306.69$^\dag$
        & 14312.23 \\
        SIS (0.01,0.005) 
        & 1579.60$^*$
        & 1587.92$^\dag$
        & 1587.012 
        & 12204.20$^*$
        & 12230.13$^\dag$
        & 12229.61 \\
        SEIR (0.01,0.005,0.05) 
        & 1052.58$^\dag$
        & 1053.95$^\dag$
        & 1051.98 
        & 8202.5$^*$
        & 8246.28$^\dag$
        & 8249.82 \\
        SEIS (0.01,0.005,0.05) 
        & 1208.57$^\dag$
        & 1211.69$^\dag$
        & 1213.51
        & 9645.5$^*$
        & 9697.68$^\dag$
        & 9684.86 \\
        SEIRct (0.01,0.005,0.05) 
        & -- 
        & 563.13$^\dag$
        & 564.60 
        & -- 
        & 4215.56$^\dag$
        & 4220.29 \\
        SEISct (0.01,0.0005,0.005) 
        & -- 
        & 838.74$^\dag$
        & 836.60
        & -- 
        & 6751.09$^\dag$
        & 6765.04 \\
        SWIR (0.01,0.005,0.05) 
        & -- 
        & 310.33$^\dag$
        & 309.41 
        & -- 
        & 2347.85$^\dag$
        & 2345.32 \\
        Threshold (0.5) $^\ddagger$
        & 2092 
        & 2092 
        & 2092 
        & 17484 
        & 17484 
        & 17484 \\
        KerteszThreshold (0.7,0.01,0.1) 
        & 2124.70$^*$
        & 2129.40$^*$
        & 2132.70 
        & 16741.30$^*$
        & 16727.60$^*$
        & 16751.35 \\
        Profile (0.5,0.01,0.1) 
        & -- 
        & 2471.40$^*$
        & 2473.68 
        & -- 
        & 18047.60$^*$
        & 18063.24 \\
        ProfileThreshold (0.7,0.25,0.01,0.1) 
        & -- 
        & 2143.70$^*$
        & 2143.61 
        & -- 
        & 15806.20$^*$
        & 15790.16 \\
        Voter 
        & 292.76
        & 293.16 
        & 293.28 
        & 2020.52 
        & 2020.83 
        & 2020.16 \\
        Qvoter (2,0.1) 
        & -- 
        & 285.76$^\dag$
        & 285.54 
        & -- 
        & 2018.51$^\dag$
        & 2017.48 \\
        MajorityRule (5) 
        & 227.34$^\dag$
        & 226.56$^\dag$
        & 227.73 
        & 1926.12$^\dag$
        & 1926.26$^\dag$
        & 1925.58 \\
        Sznajd 
        & 213.085 
        & 208.51$^\dag$
        & 208.09 
        & 1981.10$^\dag$
        & 1973.35$^\dag$
        & 1974.25 \\
        HK (0.32) $^\ddagger$
        & -- 
        & 377 
        & 377 
        & -- 
        & 2558 
        & 2558 \\
        WHK (0.32,0.2) $^\ddagger$
        & -- 
        & 882 
        & 882 
        & -- 
        & 6268 
        & 6268 \\
        \bottomrule
        \multicolumn{7}{l}{
            \footnotesize
            $\dag$: This data is the expectation of running 100 simulations.
            $^*$: This data is the expectation of running 10 simulations.
            $^\ddagger$: This model only perform one simulation.
        }
    \end{tabular}
\end{table*}

\end{document}